\documentclass[reqno,11pt,a4paper]{amsart}

\usepackage{color}
\usepackage[numbers]{natbib}
\usepackage{hyperref}

\usepackage{amsmath, amsthm, amsfonts, amssymb}
\usepackage{amscd}
\usepackage{pxfonts}

\usepackage{mathtools}
\mathtoolsset{showonlyrefs}
\numberwithin{equation}{section}
\usepackage[applemac]{inputenc}

\theoremstyle{plain}
\newtheorem{theorem}{Theorem}
\newtheorem{proposition}[theorem]{Proposition}
\newtheorem{lemma}[theorem]{Lemma}

\theoremstyle{definition}
\newtheorem{example}[theorem]{Example}

\newtheorem{definition}[theorem]{Definition}

\theoremstyle{remark}

\newtheorem{note}[theorem]{Note}

\DeclareMathOperator{\tr}{Tr}

\def\Z{\mathbb{Z}}	
\def\C{\mathbb{C}}	
\def\R{\mathbb{R}}	
\def\S{\mathcal{S}}
\newcommand\cS{{\mathcal S}}
\def\bbbn{{\mathbb N}}
\def\bbbz{{\mathbb Z}}
\newcommand\la{{\lambda}}
\def\r{{\rm r}}
\def\l{{\rm l}}
\def\c{{\rm c}}
\newcommand\cP{{\mathcal P}}
\renewcommand{\leq}{\leqslant} 		
\renewcommand{\geq}{\geqslant}

\def\bF{{\bf F}}
\def\bu{{\bf u}}

\newcommand\cR{{\mathcal R}}

\def\ud{\mathrm{d}}

\newcommand{\dev}{\partial}

\def\A{\mathcal{A}}
\def\hA{\hat{\A}}
\def\F{\mathcal{F}}
\def\hF{\hat{\F}}

\begin{document}

	\title{Recursion and Hamiltonian operators for integrable Nonabelian Difference Equations}
	\author{Matteo Casati}
	\author{Jing Ping Wang}
	\address{School of Mathematics, Statistics and Actuarial Science\\University of Kent\\Canterbury CT2 9FS, United Kingdom}
	\begin{abstract}
	In this paper, we carry out the algebraic study of integrable differential-difference equations whose field variables take values in an associative (but not commutative) algebra.
	We adapt the Hamiltonian formalism to nonabelian difference Laurent polynomials and describe how to obtain a recursion operator from the Lax representation of an integrable 
	nonabelian differential-difference system. As an application, we propose a novel family of integrable equations: the nonabelian Narita-Itoh-Bogoyavlensky lattice, for which we construct 
	their recursion operators and Hamiltonian operators and prove the locality of infinitely many commuting symmetries generated from their highly nonlocal recursion operators. Finally, 
	we discuss the nonabelian version of several integrable difference systems, including the relativistic Toda chain and Ablowitz-Ladik lattice.
	\end{abstract}
	\maketitle

	\tableofcontents

	\section{Introduction}
	This paper is devoted to the study of hidden structures for integrable differential-difference equations when their 
	field variables take their values in an associative algebra. 
	Typical examples are matrix or operator algebras. 
	Matrix versions of integrable equations appeared in the early days of the modern theory of integrable systems, when matrix KdV was introduced by Lax \cite{L68}. The main observation is that there exist many equations of this type which are integrable regardless of the dimension of the matrices \cite[Chapter 6,~7]{s2020}; even  more impressively, symmetries and conserved quantities for these systems do not depend on the single entries of these matrices, but are defined in terms of matrix polynomials. This allows us to regard the matrices (or, more precisely, the matrix-valued functions) in the systems ``simply'' as the generators of an associative but not commutative algebra: the structures we present in this paper for differential-difference systems are then independent from the matrix spaces where the problem was originally set.
	
	Nonabelian integrable systems are not only an interesting generalisation and ``abstraction'' of matrix systems; this theory can be applied to operator algebras, and the equations we obtain may be regarded as quantised version of classical integrable systems. For example, the nonabelian KdV equation can be obtained promoting the classical fields to a quantum version and replacing Poisson brackets with commutators \cite{fc95}.  
	
	Integrable partial differential evolution equations on associative algebras were studied by Olver
	and Sokolov \cite{OS98}. The authors developed the basic theory of Hamiltonian structures for associative 
	algebra-valued differential equations and presented a list of integrable one-component evolution equations of such type. The completeness of this list was proved in \cite{ow00}. In 2000, Mikhailov and Sokolov successfully brought the concepts of symmetries, first integrals, Hamiltonian and recursion 
	operators to integrable ODEs on associative algebras \cite{ms00}. These were used to prove the integrability of the Kontsevich nonabelian system \cite{WE12}. 
	
	Historically, nonabelian integrable differential-difference equations appeared as the discrete analogs of matrix integrable partial differential equations. The matrix Toda lattice is originated from a discrete version of the principal chiral field model proposed by Polyakov \cite{bmrl80, k81, mik81}. Its quantisation was studied in \cite{a}. Recently, the non-abelian Toda lattice emerged in the study of Matrix valued Hermite polynomials \cite{b}.
	Despite some specific nonabelian differential-difference equations have appeared in the literature, for example the nonabelian Volterra chain \cite{ms00} and the aforementioned Toda lattice, so far a systematic study of such equations has not been carried out.

	In Section \ref{sec:Ham} we extend the Hamiltonian 
formalism to a noncommutative (nonabelian) difference field by adapting the language introduced by Olver and Sokolov \cite{OS98} for 
	systems of PDEs to the difference case.
	Differential-difference nonabelian equations we are going to consider are evolutionary equations for functions of a discrete spatial variable and of continuous time. Their Hamiltonian description is obtained in terms of a formal (difference) calculus of variations similar to the one introduced by Kupershmidt \cite{Kup85}, but for a space of noncommutative local densities: that means that we must distinguish between operators of multiplication on the left and on the right -- that we denote $\l_a f:=af$, $\r_a f:=fa$. A local functional is written as  ``the integral of the trace of a local density''
	\begin{equation}
	F=\int\tr f.
	\end{equation}
The formal definition is thoroughly discussed in Paragraph \ref{ssec:space}; however, if our noncommutative algebra is an algebra of matrices, the trace operation is the standard one. Moreover, the operation of integral denotes the (possibly infinite) sum over all the lattice sites; in our approach we disregard the issue of convergence and identify it with the equivalence relation up to rigid shifts. 	Despite the formalism is better tailored for the local case, we use it to prove the Hamiltonian property for both local 
	and certain non-local difference operators. We give a detailed concrete example of this.
	
	In Section \ref{sec:Eqns} we begin with describing the method to construct recursion 
operators using Lax representation. This method was first proposed in \cite{mr2001h:37146} and extended to the case 
when Lax representations invariant under reduction groups \cite{wang09, kmw13,cmw19}. We then present two basic 
examples of nonabelian integrable lattices systems: the Volterra and the Toda lattice. They are obtained promoting the 
Lax representation of the corresponding commutative systems to the nonabelian case. We construct their 
	Hamiltonian and recursion operators.
	
	In Section \ref{sec:NIB} we introduce and study in detail the nonabelian version of Narita-Itoh-Bogoyavlensky 
	lattice. This is an evolutionary equation for a real-valued function $u(n,t)$, $n\in\Z, t\in\R$, such that
	\begin{equation}\label{bogadd}
	u_t=\sum_{k=1}^p \left( u_{k} u- u u_{-k}\right), \quad p\in\mathbb{N},
	\end{equation}
	where with $u_k$ (resp., $u_{-k}$) we denote the function evaluated on the translated lattice $u(n+k,t)$ (resp., $u(n-k,t)$).
	
	Equation \eqref{bogadd} is Hamiltonian for the operator
	\begin{equation}
	H=\sum_{k=1}^p\left(\r_u\cS^k\r_u-\l_u\S^{-k}\l_u-\r_u\l_u+\r_u^2+\c_u(1-\cS)^{-1}\c_u\right),
	\end{equation}
	where we have introduced the notation $\c_u:=\r_u-\l_u$ for the commutator and the shift operator $\cS^k u=u_k$. The Hamiltonian functional is
	\begin{equation}
	h=\int\tr u.
	\end{equation}
	To the best of our knowledge, this system has not appeared in the literature. We construct its
	recursion operator for arbitrary $p$ and show that this nonlocal recursion operator generates infinitely many local 
	symmetries.
	
	Finally, in Section \ref{sec:List} we present more examples of nonabelian integrable systems, obtained from their commutative counterparts studied in \cite{kmw13}.
	We give their Hamiltonian structures, Lax representations and recursion operators, which are computed using the methods in sections \ref{sec:Eqns} and \ref{sec:NIB}. 
	We believe that the equations in section \ref{RT}--\ref{BM} are new.
	
	\section{Nonabelian Hamiltonian equations: the $\theta$ formalism}\label{sec:Ham}
	In this section we present the language introduced by P.~Olver and V.~Sokolov \cite{OS98} to describe Hamiltonian equations defined on an associative (but not commutative) algebra. However, such a formalism is adapted to deal with differential difference systems, rather than partial differential ones.
	
	A differential-difference system is a set of equations describing the time evolution of $\ell$ functions of two variables $\{u^i(n,t)\}_{i=1}^\ell=\mathbf{u}(n,t)$, where $n\in\Gamma\subseteq\Z$ and $t\in\R$; In this paper we consider the infinite case $n\in\Z$. The \emph{shift operator} $\cS$ acts on the functions $\mathbf{u}(n,t)$ and gives the function evaluated on shifted variables $n$, by
	\begin{equation}
	\cS^j \mathbf{u}(n,t)=\mathbf{u}(n+j,t),
	\end{equation}
	that we denote as $\mathbf{u}_j$.
	
	As for the formal calculus of variations introduced by Kupershmidt \cite{Kup85}, we replace the functions $\{u^i(n,t)\}$ with a (formal) algebra of difference polynomials.
	\subsection{The space of nonabelian difference Laurent polynomials}\label{ssec:space}
	Let $\A$ be the linear associative algebra (over $\R$) of the Laurent polynomials in the variables $u^i_n$, $i=1,\ldots,\ell$ and $n\in\Z$. We denote the (nonabelian) product in $\A$ by simple juxtaposition; each monomial possesses an inverse by $(u^i_n)^{-1}u^i_n=u^i_n(u^i_n)^{-1}=1$ and
	$$
	(u^{i_1}_{n_1}u^{i_2}_{n_2}\cdots u^{i_p}_{n_p})^{-1}=(u^{i_p}_{n_p})^{-1}\cdots	(u^{i_2}_{n_2})^{-1}(u^{i_1}_{n_1})^{-1}.
	$$
	In this formal setting, the \emph{shift operator} is an automorphism of $\A$, namely a linear map $\cS\colon\A\to\A$ such that $\cS(ab)=\cS(a)\cS(b)$ for all $a,b\in\A$. Its action on the generators is given by
	\begin{equation}
	\cS \,u^i_n=u^i_{n+1}.
	\end{equation}
	We call the elements of $\A$ the \emph{difference (Laurent) polynomials}. Let $[a,b]=ab-ba$ denote the commutator on $\A$.
	
	Moreover, we introduce an equivalence relation in $\A$, up to cyclic permutations of the product, and we call the canonical projection the \emph{trace} --  it obviously vanishes on commutators $\tr [a,b]=0$.
	
	The basic example of $\A$ is given by regarding the generators $\{u^i_n\}$ as elements of $\mathfrak{gl}_N$; the trace is then the standard trace of matrices.
	
	The elements of the quotient space
	\begin{equation}
	\F=\frac{\A}{(\cS-1)\A+[\A,\A]}
	\end{equation}
	are called \emph{local functionals}. We denote the projection from $\A$ to $\F$ as $\int \tr \cdot$, which satisfies
	\begin{equation}
	\int\tr \cS f=\int\tr f,\qquad\qquad\int\tr fg=\int\tr gf
	\end{equation}
	for all $f,g\in\A$.
	
	Note in particular that we have a difference version of the integration by parts, namely
	\begin{equation}
	\int\tr (\cS f)g=\int\tr \cS^{-1}\left(\left(\cS f\right)g\right)=\int\tr f\left(\cS^{-1}g\right).
	\end{equation}
	\subsection{Vector fields and difference operators}
	\begin{definition}
		An \emph{evolutionary difference vector field} $X$ is a derivation of $\A$ that commutes with $\cS$.
	\end{definition}

	Derivations on $\A$ must be treated carefully because of the non-commutative product. A generic derivation on $\A$ has the form
	\begin{equation}
	X=\sum_{i=1}^\ell\sum_{n\in\Z}\phi^{i,n} \frac{\dev}{\dev u^i_n},
	\end{equation}
	$\phi^{i,n}\in\A$. Imposing the commutation rule with $\cS$ we find as a necessary and sufficient condition
	\begin{equation}
	\phi^{i,n}=\cS^n\,X^i,
	\end{equation}
	where $\{X^i\}_{i=1}^\ell$ is called the \emph{characteristic} of the evolutionary vector field.
	
	Since $X$ is by definition a derivation, the Leibniz's rule applies:
	\begin{equation}
	X(fg)=X(f)g+f X(g).
	\end{equation}
	Note that $f X(g)\neq X(g)f$; using the Leibniz's rule for monomials in $\A$ we obtain
	\begin{multline}\label{eq:evVField}
	X(u^{i_1}_{n_1}u^{i_2}_{n_2}\cdots u^{i_k}_{n_k})=\left(\cS^{n_1}X^{i_1}\right)u^{i_2}_{n_2}\cdots u^{i_k}_{n_k}+u^{i_1}_{n_1}\left(\cS^{n_2}X^{i_2}\right)\cdots u^{i_k}_{n_k}+\\
	\cdots+u^{i_1}_{n_1}u^{i_2}_{n_2}\cdots \left(\cS^{n_k}X^{i_k}\right).
	\end{multline}
	and
	\begin{equation}
	X((u^i_n)^{-1})=-u^i_n X(u^i_n) u^i_n.
	\end{equation}
	\begin{definition}
		A \emph{local scalar difference operator} is a linear map $K\colon\A\to\A$ that can be written as a linear combination 
		of terms of the form $\r_f \l_g \cS^p$ for $p\in\Z$, $f,g$ in $\A$, where $\r$ and $\l$ denote, respectively, the 
		multiplication on the right and the multiplication on the left. Namely, we have
		\begin{equation}
		\r_f\l_g\cS^p h=g \left(\cS^p h\right) f.
		\end{equation}
	\end{definition}
	The multiplication operators have the obvious properties
	\begin{equation}
	\l_f \l_g=\l_{fg}\qquad\qquad \r_f\r_g=\r_{gf}\qquad\qquad \r_f\l_g=\l_g\r_f.
	\end{equation}
	Moreover, we define the \emph{commutator} $\c_f:=\r_f-\l_f$, that is, $[f,g]=\c_g f$ and $\c_f$ is a derivation. 
	Indeed, we have
	$$
	\c_f (gh)=(\c_f g)h+g(\c_f h).
	$$
	
	The Fr\'echet derivative of $f\in\A$ is a difference operator $f_*\colon\A^\ell\to\A$ defined so that
	\begin{equation}
	f_*[P]:=\frac{\ud}{\ud \epsilon}\,f(u^i+\epsilon P^i)\Big\vert_{\epsilon=0}.
	\end{equation}
	for $P=\{P^i\}_{i=1}^\ell$, $P^i\in\A$. With the notation $f(u^i+\epsilon P^i)$ we denote the element of $\A[\epsilon]$ 
	obtained by replacing the generators $u^i$ 
	(and their shifts $u^i_n$) in the expression of $f$ with $u^i+\epsilon P^i$ (and corresponding shifts), where $\epsilon=\cS\epsilon$ is a formal constant parameter. 
	It is closely related to evolutionary vector fields, namely for any $f\in\A$ and $P$ we have $f_*[P]=X_P(f)$, where with $X_P$ we denote the evolutionary vector field with characteristic $\{P^i\}$.
	
	For a monomial $h=u^{i_1}_{n_1}u^{i_2}_{n_2}\cdots u^{i_k}_{n_k}$, we have
	\begin{eqnarray*}
		h_*[P]&=&\sum_{p=1}^k u^{i_1}_{n_1}u^{i_2}_{n_2}\cdots u^{i_{p-1}}_{n_{p-1}} \left(\cS^{n_p}P^{i_p}\right) u^{i_{p+1}}_{n_{p+1}}\cdots u^{i_k}_{n_k} \\
		&=&
		\sum_{p=1}^k \l_{u^{i_1}_{n_1}}\l_{u^{i_2}_{n_2}}\cdots \l_{ u^{i_{p-1}}_{n_{p-1}}}\r_{u^{i_{k}}_{n_{k}}}\r_{u^{i_{k-1}}_{n_{k-1}}}\cdots \r_{u^{i_{p+1}}_{n_{p+1}}}\cS^{n_p}P^{i_p}.
	\end{eqnarray*}

	The formal adjoint of the difference operator $K=\l_f \r_g \cS^p$ is
	\begin{equation}
	K^\dagger:=\cS^{-p}\r_f\l_g,
	\end{equation}
	and it is defined from the identity
	\begin{equation}
	\int \tr \left( a \,(K b)\right)=\int\tr \left( (K^\dagger a)\,b\right)
	\end{equation}
	We say that a (scalar) difference operator is \emph{skewsymmetric} if $K^\dagger=-K$.
	
	In the multi-component case, namely when $\ell>1$, we consider $\ell\times\ell$ matrices $(K)_{ij}$ whose entries are scalar difference operators. They define difference operators $K\colon\A^\ell\to\A^\ell$. 
	In this setting, the formal adjoint of $K$ is $(K^\dagger)_{ij}=(K)_{ji}^\dagger$. To avoid making the notation too heavy, we denote the entry $(K)_{ij}$ as $K^{ij}$.
	
	Given a local functional $F\in\F$, we use the Fr\'echet derivative to define its variational derivative in a fashion which is consistent with the non commutativity of the setting.
	
	Taken a generic evolutionary vector field of characteristic $v=\{v^i\}_{i=1}^\ell$, we define the variational derivative by the formula 
	\begin{equation}
	\int\tr\sum_{i=1}^\ell \frac{\delta F}{\delta u^i}v^i:=\int\tr F_*[v].
	\end{equation}
	In the commutative case, this is equivalent to the standard form
	\begin{equation}
	\frac{\delta F}{\delta u^i}=\sum_{n\in\Z}\cS^{-n}\frac{\dev F}{\dev u^i_n}.
	\end{equation}
	
	A differential-difference system
	\begin{equation}\label{eq:evsyst}
	u^i_t=X^i(\mathbf{u},\cS\mathbf{u},\cS^2\mathbf{u},\ldots) \qquad i=1,\ldots,\ell
	\end{equation}
	is identified, in this formal setting, with an evolutionary vector field of characteristic $\{X^i\}_{i=1}^\ell$. The time evolution of a local functional $F$ is, then, given by
	\begin{equation}
	\frac{\ud F}{\ud t}=\int\tr F_*[X].
	\end{equation}
	\subsection{Hamiltonian operators and functional vector fields}
	The notion of functional vector field was introduced in the context of Hamiltonian (commutative) PDEs \cite{Kup80, O80, ykb81}. It was then generalised to the nonabelian case \cite{OS98} 
	and to the commutative difference one \cite{Kup85, CW19}.
	
	In this paper we follow closely Olver and Sokolov's treatment of the subject, while a broader geometrical description will be discussed in a forthcoming work.
	
	We define basic uni-vectors $\theta_{i,n}$, where $\theta_{i,n}=\cS^n\theta_i$. Since the variables $u^i_n$ to which they are ``duals'' (more precisely, they should be regarded as dual objects of $\ud u^i_n$) 
	are not commutative with respect to the product, they do not have any particular parity under it (as opposite as in the commutative case, where they are Grassmann variables).
	
	Let $\hA:=\A[\{\theta_{i,n}\}_{i=1,n\in\Z}^\ell]$ be the space of polynomials in $\theta$'s with difference functions as coefficients. We call the elements of $\hA$ the \emph{densities of (functional) poly-vector fields}. The space $\hA$ is a graded algebra where $\deg_\theta\theta_{i,n}=1$, $\deg_\theta u^i_n=0$. Homogeneous elements of $\hA$ of degree $p$ in $\theta$ are densities of $p$-vector fields. Functional vector fields are elements of the quotient space
	\begin{equation}
	\hF=\frac{\hA}{(\cS-1)\hA+[\hA,\hA]}.
	\end{equation}
	The trace form (and as a consequence the quotient operation $\hA\twoheadrightarrow\hF$) is graded commutative, namely
	\begin{equation}\label{eq:tr}
	\tr \left(a\,b \right)=(-1)^{|a||b|}\tr \left(b\,a\right),
	\end{equation}
	where we denote $|a|:=\deg_\theta a$. To make the notation lighter, we denote $\theta_i=\theta_{i,0}$ in the multi-component case, and -- when we move to the scalar $\ell=1$ case -- $\theta_n=\theta_{1,n}$, $\theta=\theta_{1,0}$. To avoid confusion between unshifted basic univectors in the multi-component case and shifted ones in the scalar case, in the following Sections we will introduce different Latin and Greek letters denoting, respectively, different $u^i$'s and $\theta_j$'s.
	
	To any difference operator $K\colon\A^\ell\to\A^\ell$ we associate a formal vector field $\mathbf{pr}_{K\theta}$ whose characteristics is $\{\sum_{j}K^{ij}\theta_j\}$. Such a formal vector field is a graded derivation (of degree 1), namely it satisfies the graded Leibniz's property
	\begin{equation}
	\mathbf{pr}_{K\theta} (ab)=\mathbf{pr}_{K\theta}(a)b+(-1)^{|a|}a\,\mathbf{pr}_{K\theta}(b).
	\end{equation}
	
	To any difference operator $K$ we can associate the functional \emph{bivector}
	\begin{equation}\label{eq:defBiv}
	P=\frac12\int\tr\left(\sum_{i,j=1}^\ell \theta_i \,K^{ij}\theta_j\right).
	\end{equation}
	Note that the operator $K$ acts on the variable on its right. Similarly, for $K$ a \emph{skewsymmetric} difference operator we can define a bracket between local functionals
	\begin{equation}\label{eq:defBra}
	\{F,G\}:=\int\tr \left(\sum_{i,j=1}^\ell\frac{\delta F}{\delta u^i} K^{ij} \frac{\delta G}{\delta u^j}\right).
	\end{equation}
	The skewsymmetry of the bracket is equivalent to the skewsymmetry of $K$; the graded commutativity of the trace 
	and the integral in \eqref{eq:defBiv} imply that only the skewsymmetric part of $K$ is involved in the definition of $P$.
	
	\begin{definition}
		We say that a skewsymmetric difference operator $H$ is \emph{Hamiltonian} if the bracket defined via \eqref{eq:defBra} is a Lie algebra bracket on $\F$, namely it satisfies Jacobi's identity
		\begin{equation}
		\{A,\{B,C\}\}+\{B,\{C,A\}\}+\{C,\{A,B\}\}=0
		\end{equation}
		for all $A,B,C\in\F$.
	\end{definition}
	\begin{definition}
		We say that an evolutionary system \eqref{eq:evsyst} is a \emph{Hamiltonian system} if and only if
		\begin{equation}
		X^i=\sum_{j=1}^\ell H^{ij}\frac{\delta}{\delta u^j}\left(\int\tr h\right),
		\end{equation}
		with $H$ a Hamiltonian operator and for a local functional $\int\tr h$ which is called ``the Hamiltonian'' of the system.
	\end{definition}

	A statement to determine whether a differential operator is Hamiltonian is given by Olver \cite[Chapter 7]{O93} for the 
	commutative case. The same statement and proof are valid in nonabelian difference case. However, its proof 
	relies on the properties of the variational derivative and of the graded vector
	field $\mathbf{pr}_{H\theta}$. We only present the statement.
	
	\begin{theorem}
		A difference operator $H$ is Hamiltonian if and only if
		\begin{equation}\label{eq:HamProp}
		\mathbf{pr}_{H\theta}P=\frac12 \int\tr\left(\sum_{i,j=1}^\ell \mathbf{pr}_{H\theta}(\theta_i H^{ij}\theta_j)\right)=0.
		\end{equation}
		We say that the corresponding bivector $P$ is a \emph{Poisson bivector}.
	\end{theorem}

	\begin{example}
		Let $\ell=2$. We denote by $a,b$ the generators of $\A$ and by $\theta$ and $\zeta$ the corresponding basic vectors in $\hA$. The matrix operator
		\begin{equation}\label{eq:HToda-def}
		H=\begin{pmatrix}
		0 & \r_a \cS-\l_a \\ \r_a-\l_{a_{-1}}\cS^{-1} & \r_b-\l_b
		\end{pmatrix}
		\end{equation}
		is Hamiltonian.
		
		Let us denote $\Theta$ the two-component vector $(\theta,\zeta)$. We have
		\begin{equation}
		H\Theta=\begin{pmatrix}
		0 & \r_a \cS-\l_a \\ \r_a-\l_{a_{-1}}\cS^{-1} & \r_b-\l_b
		\end{pmatrix}\begin{pmatrix}
		\theta\\\zeta
		\end{pmatrix}=\begin{pmatrix}
		\zeta_1a-a\zeta\\
		\theta a-a_{-1}\theta_{-1}+\zeta b-b\zeta
		\end{pmatrix}
		\end{equation}
		and
		\begin{equation}
		P=\frac12 \int\tr\Theta_i H^{ij}\Theta_j=\int\tr\left(a\theta\zeta_1+a\zeta\theta+b\zeta\zeta\right).
		\end{equation} 
		The straightforward application of the formula for $\mathbf{pr}_{H\Theta}$, together with the graded commutativity of the trace, gives
		\begin{equation}
		\begin{split}
		\mathbf{pr}_{H\Theta}P&=\int\tr\left[\left(\zeta_1a-a\zeta\right)\left(\theta\zeta_1+\zeta\theta\right)+\left(\theta a-a_{-1}\theta_{-1}+\zeta b-b\zeta\right)\zeta\zeta\right]\\
		&=\int\tr\left(a\theta\zeta_1\zeta_1-a_{-1}\theta_{-1}\zeta\zeta\right)=\int\tr\left((\cS-1)a_{-1}\theta_{-1}\zeta\zeta\right)=0.
		\end{split}
		\end{equation}
		The Hamiltonian operator $H$ of equation \eqref{eq:HToda-def} is the local Hamiltonian structure of the nonabelian Toda system \cite{G98}, presented in Section \ref{toda}.
	\end{example}
	
	The formalism we have defined so far is tailored on \emph{local} Hamiltonian operators. However, many of the local Hamiltonian structures which have been discovered 
	in the commutative case are replaced, in the nonabelian setting, by nonlocal ones.
	For a given nonlocal operator, many of the routine computations needed to check whether it is Hamiltonian can be performed using the same technique as in the local case, relying on the properties of the integral and 
	trace operations. For instance, if the nonlocal part of an operator is of form $F (\cS-1)^{-1}G$, with $F$ and $G$ multiplication operators, namely of the form $\l_a\r_b$,
	the new element is the introduction of \emph{formal nonlocal variables}
	$$\theta':=(\cS-1)^{-1}G\theta;$$
	the parity and properties of these new variables are the same as the ones for $\theta$, but obviously we have the additional relation $(\cS-1)\theta'=G\theta$. 
	Some details about how to use these formal variables in the computations are provided in the proof of the following proposition.

	\begin{proposition}\label{thm:Volterra-Ham}The operator 
		\begin{equation}
		\label{eq:Volterra_HamStr}
			H=\r_u\cS \r_u-\l_u\cS^{-1}\l_u+\r_u\c_u-\c_u(1-\cS)^{-1}\c_u
	\end{equation}
		is Hamiltonian.
	\end{proposition}
	\begin{proof}
		We are in the scalar case, namely $\ell=1$. We denote $\theta$ the basic density of uni-vectors corresponding to the generator $u$ of $\A$.
		
		We compute the expression for the Poisson condition as in \eqref{eq:HamProp}, where $P$ is the functional bivector defined in terms of the operator $H$. Moreover, for the sake of compactness, we will split operator and bivector in local and nonlocal part. Namely, we have
		\begin{equation}
		H \theta= H_L\theta+ H_N\theta,
		\end{equation}
		with
		\begin{align}\label{eq:Volt-HL}
		H_L\theta&=\theta_1 u_1u-u u_{-1}\theta_{-1}-u\theta u+\theta u^2\\\label{eq:Volt-HN}
		H_N\theta&=(\rho-\lambda)u-u(\rho-\lambda).
		\end{align}
		We have introduced the nonlocal (densities of) uni-vectors
		\begin{equation}
		\rho=\left(\cS-1\right)^{-1}\theta u,\qquad\qquad\lambda=\left(\cS-1\right)^{-1}u\theta
		\end{equation}
		and denote $\cS^n\rho=\rho_n$, $\cS^n\lambda=\lambda_n$.
		
		The bivector defined by $H$ as $P=\frac{1}{2}\int\tr\theta H\theta$ is $P=P_L+P_N$ with
		\begin{align}
		P_L&=\frac12\int\tr\theta\left(\theta_1 u_1u-u u_{-1}\theta_{-1}-u\theta u+\theta u^2\right)\\
		&=\frac12\int\tr(\theta\theta_1 u_1u-\theta u u_{-1}\theta_{-1}-\theta u\theta u+\theta\theta u^2)\\ \label{eq:Volt-PL}
		&=\frac12\int\tr\left(2u_1u\theta\theta_1+u^2\theta\theta\right),
		\end{align}
		where we have used the trace graded commutativity to conclude that $\tr(\theta u u_{-1}\theta_{-1})=\tr(-uu_{-1}\theta_{-1}\theta)$ and $\tr(\theta u\theta u)=0$, and $\int f=\int\cS f$.
		
		For the nonlocal part, we have
		\begin{align}
		P_N&=\frac12\int\tr\left(\theta(\rho-\lambda)u-\theta u(\rho-\lambda)\right)\\
		&=\frac12\int\tr\left((\rho-\lambda)(\theta u- u\theta)\right)\\\label{eq:Volt-PN}
		&=\frac12\int\tr (\rho-\lambda)(\rho_1-\lambda_1),
		\end{align}
		where the last passage follows from the observation that
		\begin{gather}
		\theta u-u\theta= (\cS-1)\rho-(\cS-1)\lambda=\rho_1-\lambda_1-\rho+\lambda,\\
		\tr((\rho-\lambda)(\lambda-\rho))=0.
		\end{gather}
		
		The Hamiltonian condition for the operator $\mathbf{pr}_{H_1\theta}P=0$ can be, for simplicity, computed in four parts
		\begin{equation}
		\mathbf{pr}_{H_1\theta}P=\mathbf{pr}_{H_L\theta}P_L+\mathbf{pr}_{H_N\theta}P_L+\mathbf{pr}_{H_L\theta}P_N+\mathbf{pr}_{H_N\theta}P_N.
		\end{equation}
		The computation of $\mathbf{pr}_{H_L\theta}P_L$ is a straightforward application of the definition of evolutionary vector field; we use the graded commutativity on the terms we obtain to move ``most'' of the variables $u$'s on the left.
		We have
		\begin{align}
		2\mathbf{pr}_{H_L\theta}P_L&=\int\tr\left(-4u_1u\theta u\theta\theta_1+u_1^2u\theta\theta_1\theta_1+u_1u^2\theta\theta\theta_1+2u^2\theta\theta_1u_1\theta\right.\\
		&\qquad\qquad\left.-uu_1\theta_{-1}u\theta\theta+u^3\theta\theta\theta+u_1u\theta\theta u\theta_1-u^2\theta u\theta\theta\right). \label{eq:Volt-Ham1}
		\end{align}
		
		Some subtler operation must be performed when dealing with expression including the nonlocal terms. First, we compute $\mathbf{pr}_{H_N\theta}P_L$ from the definition, and exploit the graded commutativity and the properties of the integral to collect the nonlocal term $(\rho-\lambda)$ in front of the expression. The straightforward application of the formula produces terms of the form $a(\rho_1-\lambda_1)b$ that we replace with $a_{-1}(\rho-\lambda)b_{-1}$ thanks to the definition of the integral. Finally, expressions of the form $(\rho-\lambda)(a_{-1}-a)$ can be arranged as $[(\cS-1)(\rho-\lambda)]a$, which is local.
		
		We obtain
		\begin{align}
		2\mathbf{pr}_{H_N\theta}P_L&=\int\tr\left(2u_1u\theta u\theta\theta_1-2u^2\theta\theta_1u_1\theta+\right.\\
		&\qquad+\left.\left(\rho-\lambda\right)\left(2uu_{-1}\theta_{-1}\theta-2\theta\theta_1u_1u+u^2\theta\theta-\theta\theta u^2\right)\right). \label{eq:Volt-Ham2}
		\end{align}
		
		To compute the expressions $\mathbf{pr}_{H_L\theta}P_N$ and $\mathbf{pr}_{H_L\theta}P_N$, we first examine how to compute the evolution of $P_N$ along a generic vector $v$ such that $\deg_\theta v=1$. We have
		\begin{equation}
		2\mathbf{pr}_vP_N=\int\tr \mathbf{pr}_v[(\rho-\lambda)(\rho_1-\lambda_1)]
		\end{equation}
		The formal evolutionary vector field $\mathbf{pr}_v$ is a graded derivation; this fact, together with the graded commutativity, allows us to rewrite
		\begin{equation}
		\begin{split}
		2\mathbf{pr}_v P_N&=\int\tr\left[(\rho_1-\lambda_1)\mathbf{pr}_v(\rho-\lambda)-(\rho-\lambda)\mathbf{pr}_v(\rho_1-\lambda_1)\right]\\
		&=\int\tr(\rho-\lambda)\mathbf{pr}_v((\cS^{-1}-\cS)(\rho-\lambda))\\
		&=\int\tr(\rho-\lambda)\mathbf{pr}_v((1+\cS^{-1})(u\theta-\theta u))\\
		&=\int\tr(\rho-\lambda)(1+\cS^{-1})\mathbf{pr}_v(u\theta-\theta u)\\
		&=\int\tr\left[(2+\cS-1)(\rho-\lambda)\right]\mathbf{pr}_v(u\theta-\theta u)\\
		&=\int\tr2(\rho-\lambda)\mathbf{pr}_v(u\theta-\theta u)+(\theta u-u\theta)\mathbf{pr}_v(u\theta-\theta u).\\
		\end{split}
		\end{equation}
		Using this formula to compute $\mathbf{pr}_{H_L\theta}P_N$ and $\mathbf{pr}_{H_N\theta}P_N$ we obtain
		\begin{align}
		2\mathbf{pr}_{H_L\theta}P_N&=\int\tr\left[2(\rho-\lambda)\left(\theta u^2\theta-u u_{-1}\theta_{-1}\theta-u\theta u\theta-\theta u\theta u+\theta\theta u^2\right)\right.\\ \label{eq:Volt-Ham3}
		&\qquad-u_1u\theta\theta u\theta_1-u^2u_{-1}\theta_{-1}\theta\theta+2u^2\theta\theta u\theta-u^2\theta u\theta\theta\\
		&\qquad\left.-u_1u^2\theta\theta\theta_1+u u_{-1}\theta_{-1}u\theta\theta-u^3\theta\theta\theta+2u_1u\theta u\theta\theta_1\right].\\
		2\mathbf{pr}_{H_N\theta}P_N&=\int\tr\left[(\rho-\lambda)(2u\theta\theta u-\theta\theta u^2-u^2\theta\theta)\right.\\ \label{eq:Volt-Ham4a}
		&\qquad\qquad-\left.2(\rho-\lambda)(\rho-\lambda)(\theta u-u\theta)\right].
		\end{align}
		To deal with the ``double non-local term'' in \eqref{eq:Volt-Ham4a} we first observe that $\int(\cS-1)AAA=\int(A_1A_1A_1-AAA)=0$ for any element $A$. Moreover, if we denote $A_1-A=B$ we have
		\begin{equation}
		\int\tr\left((B-A)(B-A)(B-A)-AAA\right)=\int\tr(B^3+3B^2A+3A^2B)=0,
		\end{equation}
		so that we can replace in the integral and trace operation $AAB$ with $-ABB-\frac13BBB$. Performing this substitution for $A=\rho-\lambda$, $B=\theta u-u\theta$ we obtain
		\begin{equation}\label{eq:Volt-Ham4}
		\begin{split}
		2\mathbf{pr}_{H_N\theta}P_N&=\int\tr\left[(\rho-\lambda)(2\theta u\theta u-2\theta u^2\theta+2u\theta u\theta-\theta\theta u^2-u^2\theta\theta)\right.\\
		&\qquad\qquad+\left.2u^2\theta u\theta\theta-2u^2\theta\theta u\theta\right].
		\end{split}
		\end{equation}
		The summation of all the terms obtained in \eqref{eq:Volt-Ham1}, \eqref{eq:Volt-Ham2}, \eqref{eq:Volt-Ham3}, and \eqref{eq:Volt-Ham4} gives
		\begin{equation}
		2\mathbf{pr}_{H\theta}P=\int\tr\left(u_1^2u\theta\theta_1\theta_1-u^2u_{-1}\theta_{-1}\theta\theta\right)=0.
		\end{equation}
		This proves that $H$ is a Hamiltonian difference operator.
	\end{proof}
	\section{Integrable difference equations}\label{sec:Eqns}
	
	We say that an evolutionary differential-difference system 
	\begin{equation}\label{eq:pde}
	\bu_{t}=\bF(\bu_m,\bu_{m+1},\ldots \bu_{n}), \qquad  m<n \in \bbbz,\ \  \bu\in\A^l,
	\end{equation}
	is (Lax) \emph{integrable} if one can associate to \eqref{eq:pde} a pair of linear operators
	\[
	L=\cS -U(\bu;\lambda),\qquad A=D_{t}-B(\bu;\lambda),
	\]
	which is conventionally called the {\em Lax pair}. 
	
	Here $U, B$ are square matrices, whose entries are functions of the dependent variable $\bu $ and its shifts
	and certain rational (in some cases elliptic) functions of the spectral parameter $\lambda$, such that 
	equation \eqref{eq:pde} is equivalent to the compatibility of of these operators
	\begin{equation}\label{eq:zeroc}
	D_{t}(U)=\cS(B)U-UB.
	\end{equation}
	The latter is often called a {\em Lax representation} of equation 
	\eqref{eq:pde}. Systems admitting a Lax representation can be solved via the spectral transform method \cite{abl1,abl2,zak-sol}.
	
	Symmetries of an evolutionary equation are its compatible evolutionary flows.
	An integrable equation \eqref{eq:pde} has an infinite sequence of commuting symmetries
	\begin{equation}\label{eq:sym}
	\bu_{t_k}=\bF^k, \qquad k\in\bbbn\, ,
	\end{equation}
	which can be associated with a commutative algebra of linear operators 
	\begin{equation}\label{eq: LAAk}
	A^{t_k}=D_{t_k}-B^{t_k}(\bu;\lambda),\qquad [A^{t_i},\ A^{t_j}]=0 .
	\end{equation}
	Operator $A$ and equation \eqref{eq:pde} can be considered as members in the sequence of operators $\{A^{t_k}\}$ 
	and symmetries \eqref{eq:sym} respectively, for  particular values of $k$.
	
	A \emph{recursion operator} $\cR$ is a linear pseudo-difference operator mapping a symmetry of the evolutionary equation to a new symmetry. If the evolutionary equation is biHamiltonian and the first Hamiltonian structure is invertible, the operator
	$$
	\cR:=H_2 H_1^{-1}
	$$
	is a recursion operator, and $\bu_{t_k}=\cR^k \bu_{t_0}$. We call $\bu_{t_0}$ the \emph{seed} of the integrable hierarchy.
	
	In this Section, we provide the Lax representation and then compute the recursion operator for the nonabelian Volterra and Toda systems, following the general procedure detailed in Section \ref{sec:LaxGen}.

	\subsection{From the Lax representation to the recursion operator}\label{sec:LaxGen}
	In general, it is not easy to construct a recursion operator for a
	given integrable equation although we have the explicit formula for its definition. The difficulty lies in the 
	nonlocality of the operators. The case of recursion operators with a special form for the nonlocal terms, called \emph{weakly nonlocal} \cite{MaN01}, has been widely investigated; see, for example, 
	\cite{mr88g:58080,mr1974732}. 
	
	If the Lax representation of an evolutionary equation is
	known, an amazingly simple approach to
	construct a recursion operator was proposed in \cite{mr2001h:37146} and later applied for lattice equations
	\cite{Maciej200127}. 
	This idea can also be used for Lax pairs that are
	invariant under a reduction group \cite{wang09, kmw13,cmw19}. In this case, we are able to construct 
	recursion operators of form of rational difference operators \cite{cmw19}. 
	
	
	
	The main idea to derive a recursion operator using a Lax representation  is based on the fact that matrices 
	$B^{t_k}(\bu,\lambda)$ of the operators 
	$A^{t_k}=D_{t_k}-B^{t_k}(\bu,\lambda)$ corresponding to a hierarchy 
	can be related as
	\begin{equation}\label{eq:VK}
	B^{t_{k+1}}(\bu,\lambda)=\mu(\lambda)B^{t_k}(\bu,\lambda)+W^{t_k}(\bu,\lambda),
	\end{equation}
	where $\mu(\lambda)$ is a rational (elliptic in the some cases, e.g., the Landau-Lifshitz equation) multiplier and 
	$W^{t_k}(\bu,\lambda)$ is a rational matrix with a fixed (i.e. $k$ independent) divisor 
	of poles. If the system and its Lax representation is obtained as a result of a reduction with a reduction group $G$ \cite{mik80}, 
	then the multiplier  $\mu(\lambda)$ is a primitive automorphic function 
	\cite{LM05} of a finite reduction group, or  in the elliptic case is one of the  generators of the $G$-invariant 
	subring of the coordinate ring \cite{DS08}. The matrix  $W^{t_k}(\bu,\lambda)$ also depends on the dependent variables $\bu$ 
	and its shifts.
	
	The substitution of \eqref{eq:VK} in the Lax representation \eqref{eq:zeroc} results in
	\begin{equation}\label{eq:UWK}
	D_{t}(U)=\mu(\lambda) D_{\tau}(U) + \cS(W) U -U W,
	\end{equation}
	where we denote $t_{k+1}=t$, $t_k=\tau$, and $W^{t_k}=W$.
	Equation \eqref{eq:UWK} enables us to express the entries of the matrix $W$ in terms of  $\bu, \bu_{t}, 
	\bu_{\tau}$ 
	and their $\cS$-shifts. 
	It enables us to find a linear pseudo-difference operator $\cR$ such that $\bu_{t}=\cR (\bu_{\tau})$, i.e. a 
	recursion operator for a differential-difference
	hierarchy of commuting symmetries. From 
	this construction it follows that a 
	matrix $U$ and a multiplier $\mu(\lambda)$ defines a recursion operator completely and uniquely.

	In Section \ref{sec3} we present the detailed computation for the recursion operator of the Narita-Itoh-Bogoyavlensy lattice. 
	Here we show the simpler procedure for the recursion operator of the Volterra chain, which is the simplest special case of the former. For other examples, we'll only present the corresponding forms of \eqref{eq:VK}.
	\subsection{Nonabelian Volterra}\label{ssec:Volterra}
	The Nonabelian Volterra chain
	\begin{equation}\label{eq:Volterra}
	u_t=u_1u-u u_{-1}
	\end{equation}
	admits a Lax representation
	\begin{equation}
	U=\begin{pmatrix}
	\lambda & u \\ -1 & 0
	\end{pmatrix}\qquad\qquad B=\begin{pmatrix}\lambda^2+u & \lambda u\\-\lambda & u_{-1}\end{pmatrix}.
	\end{equation}
	Note that both $U$ and $B$ admit $\mathbb{Z}_2$ as a reduction group, suggesting the symmetry condition
	\begin{equation}\label{eq:symgroup}
	\begin{pmatrix} 1 & 0\\0&-1\end{pmatrix}B(-\lambda)\begin{pmatrix} 1 & 0\\0&-1\end{pmatrix}=B(\lambda).
	\end{equation}
	The ansatz for the auxiliary matrix $B^{t}$ in terms of $B^{\tau}$ must then be of the form
	\begin{equation}\label{eq:VoltAnsatz}
	\begin{split}
	B^{t}&=\lambda^2 B^{\tau}+W=\lambda^2 B^{\tau}+\lambda^2 W^{(2)} + \lambda W^{(1)} + W^{(0)}\\
	&=\lambda^2 B^{\tau}+\lambda^2\begin{pmatrix} e & 0\\ 0 & f\end{pmatrix}+\lambda \begin{pmatrix} 0 & p\\ q & 0\end{pmatrix}+\begin{pmatrix} s & 0\\ 0 & t\end{pmatrix};
	\end{split}
	\end{equation}
	the symmetry of the matrices $W^{(s)}$ is determined by condition \eqref{eq:symgroup}.
	
	From \eqref{eq:UWK}, we have
	$$
	\begin{pmatrix}
	0 & u_t \\ 0 & 0
	\end{pmatrix}=\lambda^2 \begin{pmatrix}0 & u_\tau\\0 & 0\end{pmatrix}+\cS(W)U-UW,
	$$
	which leads to the system of equations (here $p_1=\cS p$)
	\begin{align}
	\lambda^2\colon&&0&=e_1-e\\
	&&0&=-f_1+q_1\\
	&&u_\tau&=p+u f-e_1u\\
	\lambda^1\colon&&0&=s_1-s-u q-p_1\\
	&&0&=q_1 u+p\\
	\lambda^0\colon&&0&=-t_1+s\\
	&&u_t&=s_1u-u t
	\end{align}
	The first equation implies that $e$ must be a constant, that we set equal to 0. Our aim is to obtain the operator $\Re$ such that $u_t=\Re u_\tau$.
	
	The system can be solved observing that
	\begin{gather}
	u_t=(\r_u\cS-\l_u\cS^{-1})s,\\
	(\cS-1)s=(\l_u-\cS \r_u\cS)q,\\
	(\l_u-\r_u\cS) q=u_\tau,
	\end{gather}
	from which the recursion operator reads
	\begin{equation}\label{eq:recVolterra}
	\Re=\left(\r_u\cS-\l_u\cS^{-1}\right)\left(\cS-1\right)^{-1}\left(\l_u-\cS \r_u \cS\right)\left(\l_u-\r_u\cS\right)^{-1}.
	\end{equation}

	The nonabelian Volterra chain is Hamiltonian with respect to the Hamiltonian difference operator \eqref{eq:Volterra_HamStr}	and the Hamiltonian functional
	\begin{equation}\label{eq:VoltHamF}
	h_0=\int\tr u.
	\end{equation}

	\subsection{Nonabelian Toda}\label{toda}
	\label{sec:Toda}
	
	The Nonabelian Toda system
	\begin{equation}\label{eq:naToda}
	\begin{cases}
	a_t =& b_1 a - a b\\
	b_t =& a - a_{-1}
	\end{cases}
	\end{equation}
	is Hamiltonian with respect to the Hamiltonian difference operator (cf. \eqref{eq:HToda-def}) 
	\begin{equation}\label{eq:H1}
	H_1=\begin{pmatrix}
	0 & \r_a \cS-\l_a \\ \r_a-\l_{a_{-1}}\cS^{-1} & \r_b-\l_b
	\end{pmatrix}
	\end{equation}
	and the Hamiltonian functional
	\begin{equation}
	h_0=\int\tr\left(a+\frac12 b^2\right).
	\end{equation}
	System \eqref{eq:naToda} admits a Lax pair representation given by
	\begin{equation}\label{eq:TodaLax}
	U=\begin{pmatrix}\lambda+b_1 & a\\-1&0\end{pmatrix}\qquad\qquad B=\begin{pmatrix}0 & -a\\1 & \lambda+b\end{pmatrix}.
	\end{equation}
	Take as the ansatz
	\begin{equation}\label{eq:TodaAnsatz}
	B^{t}=\lambda B^{\tau}+\lambda W^{(1)}+W^{(0)}
	\end{equation}
	with $W^{(1)}$ and $W^{(0)}$ generic $2\times2$ matrices;	using \eqref{eq:UWK}, we get
	the recursion operator $\Re$ for the system
	\begin{equation}
	\Re=\begin{pmatrix}
	\l_{b_1}\r_a\cS-\l_a\r_b & \r_a\cS^2-\l_a\\
	\r_a\cS-\l_{a_{-1}}\cS^{-1} & \r_b \cS-\l_b
	\end{pmatrix}
	\begin{pmatrix}
	\left(\l_a-\r_a\cS\right)^{-1} & 0\\0 & (1-\cS)^{-1}
	\end{pmatrix}.
	\end{equation}
	The second Hamiltonian structure $H_2$, obtained by $H_2=\Re H_1$, is given by a matrix with the following entries \cite{li}:
	\begin{align}
	\left(H_2\right)_{11}&=\l_a \cS^{-1}\l_a-\r_a\cS \r_a-\r_a\c_a+\c_a(1-\cS)^{-1}\c_a\\
	\left(H_2\right)_{12}&=\l_a\r_b-\r_a\c_b-\r_a\cS \r_b+\c_a(1-\cS)^{-1}\c_b\\
	\left(H_2\right)_{21}&=-\r_{ab}+\l_b\cS^{-1}\l_a+\c_b(1-\cS)^{-1}\c_a\\
	\left(H_2\right)_{22}&=\cS^{-1}\l_a-\r_a\cS-\r_b\c_b+\c_b(1-\cS)^{-1}\c_b.
	\end{align}
	
	The system \eqref{eq:naToda} is Hamiltonian with respect to $H_2$ and the Hamiltonian functional $h_{-1}=-\int\tr b$, which is a Casimir of $H_1$.
	Moreover, the constant change of coordinates $b\mapsto b+\eta$ induces for $H_2$ the transformation
	\begin{equation}
	H_2\mapsto H_2-\eta H_1.
	\end{equation}
	This means that $(H_1,H_2)$ form a biHamiltonian pair.
	
	\begin{note}The \emph{nonabelian two-component Volterra} equation
		\begin{eqnarray*}
			\left\{ {\begin{array}{l} u_t= v_1u-uv\\ v_t=  uv-vu_{-1} \end{array} } \right.
		\end{eqnarray*}
		can be obtained from Volterra chain $w_t=w_1w-ww_{-1}$ introducing $u$ and $v$ as the even and odd elements of the lattice $u(n,t):=w(2n,t)$, $v(n,t):=w(2n-1,t)$, or from the Toda lattice with the Miura transformation
		\begin{equation}
		a=uv\qquad\qquad b=u_{-1}+v.
		\end{equation}
		A Lax pair for the system is
		\begin{eqnarray*}
			L= \lambda \cS-u-v_{1}+\lambda^{-1}u v \cS^{-1}; \qquad
			A= \lambda^{-1}uv\cS^{-1}.
		\end{eqnarray*}
	\end{note}
	
	\section{Integrable nonabelian Narita-Itoh-Bogoyavlensky lattice}\label{sec:NIB}
	In this section, we consider a family of nonabelian differential difference equations \eqref{bogadd}. It is the nonabelian version of the
	Narita-Itoh-Bogoyavlensky lattice \cite{Bog88, narita, itoh} known as an integrable
	discretisation for the Korteweg-de Vries equation:
	\begin{equation}\label{bogcom}
	u_t=u \sum_{k=1}^p \left( u_{k} - u_{-k}\right).
	\end{equation}
	For the Narita-Itoh-Bogoyavlensky lattice of arbitrary $p\in\bbbn$, its recursion operator and Hamiltonian structures have been extensively studied
	in \cite{wang12}. Note that for $p=1$ the equation reduces to the equation for the Volterra chain (see Section \ref{ssec:Volterra}). Remarkably there are explicit expressions for all its symmetries. 
	In this section, we show that the method used for the commutative case can be extended to the nonabelian case in a straightforward manner. We are going to present their recursion operators and 
	to prove that these operators generate local symmetries.
	
	\subsection{Construction of recursion operators}\label{sec3}
	
	In this section, we use the ideas presented in Section \ref{sec:LaxGen} to construct a recursion operator of system \eqref{bogadd} from
	its Lax representation. 
	
	The Lax operator of \eqref{bogadd} for any $p\in\bbbn$ is given in \cite{Bog88} as follows:
	\begin{equation}\label{lax0}
	L=\cS+u\cS^{-p}\ . 
	\end{equation}
	We can compute its hierarchy of symmetry flows using the formula
	\begin{eqnarray*}
		&&L_{t_k}=[B^{(k)},\ L], \qquad B^{(k)}=(L^{(p+1)k})_{\geq 0},
	\end{eqnarray*}
	where $\geq 0$ means taking the terms with non-negative power of $\cS$ in $L^{(p+1)k}$.
	
	We rewrite \eqref{lax0} into the following matrix form
	\begin{equation}\label{lax}
	L(\la)=\cS-\la U^{(0)} - U^{(1)}:=\cS-C(\la) ,
	\end{equation}
	where $U^{(0)}=(u^{(0)}_{i,j})$ and $U^{(1)}=(u^{(1)}_{i,j})$ are $(p+1)\times
	(p+1)$ matrices.  The matrices $U^{(0)}$ and $U^{(1)}$ are given by
	\begin{eqnarray*}
		U^{(0)}  =\left(\begin{array}{lcrl} 1 &0 & \cdots&  0 \\ 0 & 0 & \cdots
			& 0\\
			\vdots& \vdots & \ddots &\vdots \\0&0&\cdots&0 \\0&0&\cdots&  0
		\end{array} \right) 
		\quad \mbox{and}\quad
		U^{(1)}  =\left(\begin{array}{lcrll} 0 &0 & \cdots& 0 & -u \\ 1 & 0 & \cdots & 0
			& 0\\
			\vdots & \ddots& &\vdots&\vdots \\0&\cdots&1&0&0 \\0&\cdots& 0& 1&0
		\end{array} \right) 
	\end{eqnarray*}
	respectively, that is, their non-zero entries are
	$$ u^{(0)}_{11}=1, \quad u^{(1)}_{1,p+1}=-u \quad \mbox{and}\quad u^{(1)}_{i+1,i}=1, \ \ i=1,2,\cdots,
	p.$$
	We rewrite the operator $B$ in a matrix form denoted by $B(\la)$, so that the symmetry flows can be obtained by the zero curvature equation
	\begin{equation}\label{zero}
	C(\la)_t=\cS(B(\la))\ C(\la)-C(\la) B(\la) .
	\end{equation}
	Make the Ansatz
	\begin{equation}\label{mw}
	\bar B(\la)=\la^{p+1} B(\la)+ W(\la),\quad \mbox{where} \quad W(\la)=\sum_{i=0}^{p+1}\la^{p+1-i}A^{(i)},
	\end{equation}
	where $A^{(i)}=(a^{(i)}_{kl})$ are $(p+1)\times (p+1)$ matrices 
	with the only non-zero entries being $a^{(i)}_{j+i,j}$ for $1\leq j\leq p+1$.
	Here we read $i+j$ as $(i+j) \mod (p+1)$. For simplicity, we shall continue to
	denote the index $l$ when $l>p+1$ instead of $l\mod (p+1)$. So both $A^{(0)}$ and $A^{(p+1)}$ are diagonal 
	matrices.
	
	The Ansatz $W$ is invariant under the following transformation
	$$
	r: W(\la)\mapsto P W(\sigma \la) P^{-1} ,
	$$
	where $P$ is a diagonal $(p+1)\times (p+1)$ matrix given by $P_{ii}=\sigma^{i}$ and
	$\sigma=e^{2\pi i/(p+1)}$ since we have $P^{-1} A^{(i)} P=\sigma^i A^{(i)}$. The transformation satisfies
	$r^{p+1} = id $.
	
	The zero curvature condition \eqref{zero} for $\bar B(\la)$ leads to
	a formula for computing a recursion operator as follows:
	\begin{equation}\label{rel}
	C_{t_{n+1}}=\cS(\bar B(\la))\ C(\la)-C(\la) \bar B(\la) =\la^{p+1} C_{t_n} + \cS(W) C-C W.
	\end{equation}
	Substituting \eqref{lax} and \eqref{mw} into \eqref{rel} and collecting the
	coefficients of this $\la$-polynomial, we obtain
	\begin{align}
	&\la^{p+2}:\quad     \cS(A^{(0)}) U^{(0)}-U^{(0)} A^{(0)} =0; \label{a0}\\
	&\la^{p+1}: \quad   U^{(1)} _{t_n}+\cS(A^{(1)}) U^{(0)}-U^{(0)}
	A^{(1)}+\cS(A^{(0)})
	U^{(1)}-U^{(1)} A^{(0)}=0;\label{a1}\\
	&\la^{p+1-i}:   \cS(A^{(i+1)}) U^{(0)}-U^{(0)} A^{(i+1)}+\cS(A^{(i)})
	U^{(1)}-U^{(1)} A^{(i)}=0, 1\leq i\leq p;\nonumber\\
	&\la^{0}: \qquad    U^{(1)}_{t_{n+1}}=\cS(A^{(p+1)}) U^{(1)}-U^{(1)}
	A^{(p+1)}.\label{ap}
	\end{align}
	\begin{lemma}\label{lem1}
		Assume that $p\geq 2$ and $ 1\leq i\leq p-1$. The matrix equations 
		\begin{equation}\label{matrixeq}
		\cS(A^{(i+1)}) U^{(0)}-U^{(0)} A^{(i+1)}+\cS(A^{(i)}) U^{(1)}-U^{(1)} A^{(i)}=0 
		\end{equation}
		are equivalent to 
		\begin{align}\label{reai1}
		&\cS(a^{(i+1)}_{i+2,1})+\cS(a^{(i)}_{i+2,2})-a_{1+i,1}^{(i)}=0
		;\\\label{reai2}
		&-a^{(i+1)}_{1,p+1-i}+\cS(a^{(i)}_{1,p+2-i})+\l_u
		a_{p+1,p+1-i}^{(i)}=0;\\\label{reai3}
		&-\r_u \cS(a^{(i)}_{i+1,1})-a_{i,p+1}^{(i)}=0;\\\label{reai4}
		&\cS(a^{(i)}_{l+i+1,l+1})-a_{l+i,l}^{(i)}=0,\quad 2\leq l\leq
		p,\quad l\neq p+1-i \ .
		\end{align}
	\end{lemma}
	\begin{proof}. We directly compute the multiplications of matrices and write
		out their non-zero entries, respectively:
		\begin{eqnarray*}
			&&( \cS(A^{(i+1)}) U^{(0)})_{i+2,1}=\cS(a^{(i+1)}_{i+2,1})
			u^{(0)}_{11}=\cS(a^{(i+1)}_{i+2,1});\\
			&&(U^{(0)} A^{(i+1)})_{1,p+1-i}=u^{(0)}_{11}
			a^{(i+1)}_{1,p+1-i}=a^{(i+1)}_{1,p+1-i};\\
			&&(\cS(A^{(i)}) U^{(1)})_{i+1,p+1}=\cS(a^{(i)}_{i+1,1}) u_{1,p+1}^{(1)}=-\r_u
			\cS(a^{(i)}_{i+1,1});\\
			&&(\cS(A^{(i)}) U^{(1)})_{l+i+1,l}=\cS(a^{(i)}_{l+i+1,l+1})
			u_{l+1,l}^{(1)}=\cS(a^{(i)}_{l+i+1,l+1}),\quad 1\leq l\leq p;\\
			&&(U^{(1)} A^{(i)})_{1,p+1-i}=u_{1,p+1}^{(1)} a_{p+1,p+1-i}^{(i)}=-\l_u
			a_{p+1,p+1-i}^{(i)} ;\\
			&&(U^{(1)} A^{(i)})_{l+i+1,l}=u_{l+i+1,l+i}^{(1)}
			a_{l+i,l}^{(i)}=a_{l+i,l}^{(i)},\quad 1\leq l\leq p+1, \quad l\neq p+1-i;
		\end{eqnarray*}
		We are now ready to write out the entries for the matrix equations, which leads to the formulas
		stated in the lemma.
	\end{proof}
	
	Specifically, when $i=0$, using the proof of Lemma \ref{lem1} we obtain the equivalent conditions for 
	the matrix equation \eqref{matrixeq} as follows:
	\begin{align}
	&\cS(a^{(1)}_{2,1})+\cS(a^{(0)}_{2,2})-a_{1,1}^{(0)}=0;
	\label{rea00}\\
	&-a^{(1)}_{1,p+1}+\cS(a^{(0)}_{1,1})+\l_u
	a_{p+1,p+1}^{(0)}=0;\label{rea01}\\
	&\cS(a^{(0)}_{l+1,l+1})-a_{l,l}^{(0)}=0,\quad 2\leq l\leq
	p \ .\label{rea02}
	\end{align}
	In the similar way, the matrix equation \eqref{matrixeq} for $i=p$ is equivalent to
	\begin{align}
	&\cS(a^{(p+1)}_{1,1})-a^{(p+1)}_{1,1}+\cS(a^{(p)}_{1,2})+\l_u
	a_{p+1,1}^{(p)}=0;\label{reap0}\\
	&-\r_u \cS(a^{(p)}_{p+1,1})-a_{p,p+1}^{(p)}=0;\label{reap1}\\
	&\cS(a^{(p)}_{l,l+1})-a_{l-1,l}^{(p)}=0,\quad 2\leq l\leq p \ .\label{reap2}
	\end{align}
	Notice that formula \eqref{rea00}, \eqref{rea01}, \eqref{reap0} and \eqref{reap1} are valid for $p=1$.
	
	Using formula \eqref{reai1}--\eqref{reai4} in Lemma \ref{lem1}, we can now find the relation between $a^{(i+1)}_{i+2,1}$ 
	and $a^{(i)}_{i+1,1}$.
	\begin{lemma}\label{lem2}
		Assume that $p\geq 2$ and $1\leq i \leq p-1$. We have
		\begin{equation}\label{keyre}
		a^{(i+1)}_{i+2,1}=\cS^{-1} (\cS^{i} \r_u-\l_u \cS^{i-p})^{-1} (\cS^{i} \r_u  -\l_u
		\cS^{i-p-1}) \cS  (a^{(i)}_{i+1,1}) .
		\end{equation}
	\end{lemma}
	\begin{proof}First using formula \eqref{reai4}, we can show that 
		\begin{equation}\label{a1pi}
		a_{1,p+2-i}^{(i)}=\cS^{i-1} (a_{i,p+1}^{(i)}) .
		\end{equation}
		Indeed, if we take $l=p+2-i+r$ in \eqref{reai4}, it follows that
		$$a_{1+r,p+2-i+r}^{(i)}=\cS(a^{(i)}_{2+r,p+3-i+r}), \quad 0\leq r \leq i-2 .
		$$
		Thus we recursively obtain \eqref{a1pi}, that is,
		$$a_{1,p+2-i}^{(i)}=\cS(a^{(i)}_{2,p+3-i})=\cdots =\cS^{i-1} (a_{i,p+1}^{(i)}) .
		$$ 
		From \eqref{reai3}, it leads to $a_{i,p+1}^{(i)}=-\r_u \cS(a^{(i)}_{i+1,1})$. Substituting it into to \eqref{a1pi}, we get 
		\begin{equation}\label{sai}
		a_{1,p+2-i}^{(i)} =-\cS^{i-1} \r_u \cS(a^{(i)}_{i+1,1}). 
		\end{equation}
		Similar as the proof of \eqref{a1pi}, by taking $ l=2+r$, where $0\leq r\leq p-2-i$ in \eqref{reai4} when $p\geq 2$, we recursively show that
		\begin{equation}\label{saip}
		a_{2+i,2}^{(i)} =\cS(a^{(i)}_{3+i,3})=\cdots =\cS^{p-1-i}
		(a_{p+1,p+1-i}^{(i)}).
		\end{equation}
		Now we substitute \eqref{saip} into \eqref{reai1} and it leads to
		\begin{equation}\label{pi}
		\cS(a^{(i+1)}_{i+2,1})+\cS^{p-i} (a_{p+1,p+1-i}^{(i)})-a_{1+i,1}^{(i)}=0.
		\end{equation}
		Formula \eqref{sai} is valid for different values of $1\leq i\leq p-1$. Thus we have
		$a_{1,p+1-i}^{(i+1)} =-\cS^{i} \r_u \cS(a^{(i+1)}_{i+2,1})$.  We substitute it and \eqref{sai} into \eqref{reai2} and it becomes 
		\begin{equation}\label{pii}
		\cS^{i} \r_u \cS(a^{(i+1)}_{i+2,1})-\cS^{i} \r_u
		\cS(a^{(i)}_{i+1,1})+\l_u a_{p+1,p+1-i}^{(i)}=0 .
		\end{equation}
		From \eqref{pi} and \eqref{pii}, we eliminate $a_{p+1,p+1-i}^{(i)}$ and obtain
		the relation between $a^{(i+1)}_{i+2,1}$ and $a^{(i)}_{i+1,1}$, that is,
		\begin{eqnarray*}
			a^{(i+1)}_{i+2,1}=(\cS^{i} \r_u \cS-\l_u \cS^{i+1-p})^{-1} (\cS^{i} \r_u \cS -\l_u
			\cS^{i-p})   (a^{(i)}_{i+1,1}),
		\end{eqnarray*}
		which is equivalent to \eqref{keyre} as written in this lemma.
	\end{proof}
	
	From \eqref{reap1}) and \eqref{reap2}, we can see that \eqref{sai} is also valid  for $i=p$. Substituting it into
	\eqref{reap0}, we obtain 
	$$(\cS-1)a^{(p+1)}_{1,1}=-\cS(a^{(p)}_{1,2})-\l_u a_{p+1,1}^{(p)} =(\cS^{p} \r_u
	\cS -\l_u) (a^{(p)}_{p+1,1}).
	$$
	Therefore,
	\begin{equation}\label{keyp}
	a^{(p+1)}_{1,1}=(\cS-1)^{-1} (\cS^{p} \r_u -\l_u\cS^{-1})\cS (a^{(p)}_{p+1,1}).
	\end{equation}
	
	We are now ready to compute the recursion operators for any $p\in\bbbn$.
	\begin{theorem}\label{th1}
		A recursion operator of Narita-Itoh-Bogoyavlensky lattice 
		\eqref{bogadd} is
		\begin{equation}\label{Re}
		\Re=(\r_u \cS -\l_u \cS^{-p}) (\cS -1)^{-1} \prod_{i=1}^{\rightarrow{p}} (\cS^{p+1-i} \r_u-\l_u\cS^{-i} )
		(\cS^{p-i} \r_u-\l_u \cS^{-i})^{-1} \ .
		\end{equation}
	\end{theorem}
	Since the difference operators are not commuting, here we use the notation 
	$\prod_{i=1}^{\rightarrow{p}}$ to denote the order of the value $i$, from $1$ to $p$, that is,
	$\prod_{i=1}^{\rightarrow{p}}a_i=a_1 a_2 \cdots a_p$ and
	$\prod_{i=1}^{\leftarrow{p}}a_i=a_p a_{p-1} \cdots a_1$. Later we use the same notation for difference polynomials.
	
	\begin{proof} First from \eqref{rea02}, we can show that
		\begin{equation}\label{A0}
		a^{(0)}_{ll}=\cS^{p+1-l}(a^{(0)}_{p+1,p+1}),\quad 2\leq l \leq p+1 .
		\end{equation}
		The next identity \eqref{a0}) leads to $(\cS-1)a^{(0)}_{11}=0$. Here we choose the solution
		$a^{(0)}_{11}=0$, which makes it possible to find the relation between $u_{t_{n+1}}$ and $u_{t_n}$ .
		We now substitute this into \eqref{rea00} and using \eqref{A0} we obtain
		\begin{equation}\label{A1p1}
		a^{(1)}_{2,1}=-a^{(0)}_{2,2}=-\cS^{p-1}(a^{(0)}_{p+1,p+1}) .
		\end{equation}
		It follows from \eqref{a1} and \eqref{rea01} that
		\begin{eqnarray*}
			-u_{t_n}-a^{(1)}_{1,p+1}+\l_u a_{p+1,p+1}^{(0)}=0 .
		\end{eqnarray*}
		We now show that $a^{(1)}_{1,p+1}=-u \cS (a^{(1)}_{2,1})$ for all $p\in\bbbn$. When $p\geq 2$, 
		we have $a^{(1)}_{1,p+1}=-\r_u \cS (a^{(1)}_{2,1})$ by taking $i=1$ in \eqref{reai3}. With $p=1$, this is a result from  \eqref{reap1}. Therefore for all $p\in \bbbn$, using \eqref{A1p1} we have
		$a^{(1)}_{1,p+1}=\r_u \cS^p (a^{(0)}_{p+1,p+1})$ . This leads to
		\begin{eqnarray*}
			u_{t_n}+(\r_u \cS^{p}-\l_u) a_{p+1,p+1}^{(0)}=0.
		\end{eqnarray*}
		Therefore, we have
		\begin{equation}\label{ret}
		a_{p+1,p+1}^{(0)}=(\l_u-\r_u \cS^{p})^{-1} u_{t_n} \ .
		\end{equation}
		From \eqref{ap}, we find that
		\begin{eqnarray*}
			\left\{\begin{array}{l}
				-u_{t_{n+1}}=-\r_u\cS(a^{(p+1)}_{1,1})+\l_u a_{p+1,p+1}^{(p+1)};\\
				\cS(a^{(p+1)}_{l+1,l+1})=a_{l,l}^{(p+1)},\quad 1\leq l\leq p \ .\end{array}\right.
		\end{eqnarray*}
		This implies that
		\begin{eqnarray*}
			u_{t_{n+1}}= (\r_u \cS -\l_u \cS^{-p}) (a^{(p+1)}_{1,1}).
		\end{eqnarray*}
		Using \eqref{keyp} we rewrite it as
		\begin{eqnarray*}
			u_{t_{n+1}}= (\r_u \cS -\l_u \cS^{-p}) (\cS-1)^{-1} (\cS^{p} \r_u -\l_u\cS^{-1} )\cS
			(a^{(p)}_{p+1,1}). 
		\end{eqnarray*}
		Using Lemma \ref{lem2} when $p\geq 2$, we obtain
		\begin{eqnarray*}
			u_{t_{n+1}}=(\r_u \cS\! -\!\l_u \cS^{-p}) (\cS\!-\!1)^{-1} (\cS^{p} \r_u \! -\!\l_u\cS^{-1})
			\prod_{i=1}^{\rightarrow{(p-1)}} (\cS^{p-i} \r_u\!-\!\l_u \cS^{-i})^{-1} (\cS^{p-i} \r_u \! -\!\l_u
			\cS^{-i-1}) \cS (a^{(1)}_{2,1})\ .
		\end{eqnarray*}
		This is also valid for $p=1$ under the convention that the empty product is equal to $1$.
		Finally, using \eqref{A1p1} and \eqref{ret}, we get the relation between
		$u_{t_{n+1}}$ and $u_{t_n}$, which gives rise to the recursion operator as
		stated in this theorem.
		%
	\end{proof}
	
	When $p=1$, it follows from Theorem \ref{th1} that
	\begin{eqnarray*}
		\Re=(\r_u \cS -\l_u \cS^{-1}) (\cS -1)^{-1}  (\cS \r_u-\l_u\cS^{-1} )
		( \r_u-\l_u \cS^{-1})^{-1} \ ,
	\end{eqnarray*}
	which is the recursion operator for the Nonabelian Volterra chain we have already computed in Section \ref{ssec:Volterra} (compare with \eqref{eq:recVolterra}).

	\subsection{Locality of symmetries}\label{sec4}
	The recursion operators we have obtained in Theorem \ref{th1} have nonlocal terms. One important question
	is whether the operator is guaranteed to generate local symmetries
	starting from a proper seed. 
	Sufficient conditions for weakly nonlocal Nijenhuis
	differential operators are formulated in \cite{mr1974732, serg5},
	which are also valid for  weakly nonlocal Nijenhuis
	difference operators \cite{mwx2}. This result is generalised to Nijenhuis
	operators, which are the product of weakly nonlocal Hamiltonian 
	and symplectic operators \cite{wang09}. 
	However, the recursion operator given in Theorem \ref{th1} is not
	weakly nonlocal. In this section, we are going to directly prove the
	locality of symmetries by induction. 
	
	To do so, we first introduce a family of homogeneous difference polynomials of
	degree $l$ with respect to  the dependent  variable $u$ and its shifts
	\begin{equation}\label{sfunc}
	\cP^{(l,k)}=\sum_{0\leq \lambda_{l-1}\leq \cdots \leq \lambda_0\leq k}
	\left(\prod_{j=0}^{\leftarrow{(l-1)}} u_{\lambda_j+j p} \right),
	\end{equation}
	where $k\geq 0, l\geq 1$ and $p\geq 1$ are all integers. In particular, for
	any $p\in\bbbn$ we have
	\begin{equation}\label{sfuncp}
	\cP^{(1,k)}= \sum_{j=0}^k u_j \quad \mbox{and} \quad \cP^{(l,0)}=u_{(l-1)p} u_{(l-2)p} \cdots u_p u \ .
	\end{equation}
	Equation \eqref{bogadd} can be written in terms of these polynomials as follows:
	\begin{equation}\label{bogfun}
	u_t=(\r_u \cS-\l_u\cS^{-p})\cP^{(1,p-1)} .
	\end{equation}
	
	This family of polynomials was first defined in \cite{svin09}, where the author
	amazingly gave the explicit expressions for the  hierarchy of symmetries of equation
	\eqref{bogadd} in terms of it. We are going to show that the recursion operator
	\eqref{Re} generates the same symmetries as in this paper modulo a sign. 
	
	It is easy to see that the polynomials \eqref{sfunc} possess the following
	properties \cite{svin09,svin11}:
	\begin{align}
	&\cP^{(l,k)}-\cP^{(l,k-1)}=\r_{u_k} \cS^p (\cP^{(l-1,k)});\label{prop1}\\
	&\cP^{(l,k)}-\cS(\cP^{(l,k-1)})=\l_{u_{(l-1)p}} \cP^{(l-1,k)}.\label{prop2}
	\end{align}
	These immediately lead to 
	\begin{equation}
	(\cS-1) \cP^{(l,k)}=\left( \cS \r_{u_k} \cS^p-\l_{u_{(l-1)p}}\right) \cP^{(l-1,k)}.\label{prop3}
	\end{equation}
	We now prove another important property.
	\begin{proposition}\label{pro}
		For all $l, p\in \bbbn$, we have
		\begin{equation}\label{imp}
		(\cS^{p-i} \r_u-\l_u \cS^{-i})\cS^{-lp+i} \cP^{(l,(l+1)p-i)}= (\cS^{p-i}
		\r_u-\l_u\cS^{-(i+1)} ) \cS^{-lp+i+1} \cP^{(l,(l+1)p-i-1)}, \ \ 0\leq i \leq p .
		\end{equation}
	\end{proposition}
	\begin{proof}. Let us compute the difference between the left-hand side and the
		right-hand side of the identity \eqref{imp} using the properties
		\eqref{prop1} and \eqref{prop2}:
		\begin{eqnarray*} 
			&&\quad (\cS^{p-i} \r_u-\l_u \cS^{-i})\cS^{-lp+i} \cP^{(l,(l+1)p-i)}- (\cS^{p-i}
			\r_u-\l_u\cS^{-(i+1)} ) \cS^{-lp+i+1} \cP^{(l,(l+1)p-i-1)}\\
			&&=\r_{u_{p-i}} \cS^{-lp +p} \left( \cP^{(l,(l+1)p-i)} -\cS
			\cP^{(l,(l+1)p-i-1)} \right)-\l_u \cS^{-lp} \left(  \cP^{(l,(l+1)p-i)}- 
			\cP^{(l,(l+1)p-i-1)}\right)\\
			&&=\r_{u_{p-i}} \cS^{-lp +p} \left( \l_{u_{(l-1)p}} \cP^{(l-1,(l+1)p-i)} \right) -\l_u
			\cS^{-lp} \left( \r_{u_{(l+1)p-i}}\cS^p \cP^{(l-1,(l+1)p-i)}\right)\\
			&&=\r_{u_{p-i}} \l_u \cS^{-lp +p} \cP^{(l-1,(l+1)p-i)} -\l_u \r_{u_{p-i}}
			\cS^{-lp+p}  \cP^{(l-1,(l+1)p-i)}=0 .
		\end{eqnarray*}
		We proved the statement.
	\end{proof}
	
	Notice that we can rewrite the recursion operator \eqref{Re} in the form
	\begin{align}
	&\Re=(\r_u \cS -\l_u \cS^{-p}) (\cS -1)^{-1}(\cS^{p} \r_u-\l_u\cS^{-1} ) \cdot \!\! \prod_{i=1}^{\rightarrow{(p-1)}} (\cS^{p-i} \r_u-\l_u\cS^{-i} )^{-1}
	(\cS^{p-i} \r_u-\l_u \cS^{-(i+1)}) \nonumber\\
	&\quad \cdot ( \r_u-\l_u \cS^{-p})^{-1} \ . \label{rere}
	\end{align}
	Using it, we are able to prove the following result:
	\begin{theorem} For Narita-Itoh-Bogoyavlensky lattice \eqref{bogadd}, starting
		from the equation itself, its symmetries $Q^l=\Re^l(u_t)$ generated
		by recursion operator \eqref{Re} are local and
		\begin{equation}\label{syms}
		Q^l=\Re^l(u_t)=(\r_u \cS-\l_u \cS^{-p}) \cS^{-lp} \cP^{(l+1,(l+1)p-1)}\quad \mbox{for
			all} \quad   0\leq l\in\bbbz .
		\end{equation}
	\end{theorem}
	\begin{proof} First the statement is clearly true for $l=0 $ from \eqref{bogfun}.
		Assume the statement is true for $l-1\geq 0$. Let us compute the next symmetry
		$Q^{l}$.
		Taking $i=p$  in \eqref{imp} and using the induction assumption, we have
		\begin{eqnarray*}
			( r_u-\l_u \cS^{-p})\cS^{-lp+p} \cP^{(l,lp)}= (\r_u-\l_u\cS^{-(p+1)} ) \cS^{-lp+p+1}
			\cP^{(l,lp-1)}=Q^{l-1}.
		\end{eqnarray*}
		Hence
		\begin{eqnarray*}
			&&Q^l=\Re (Q^{l-1})=(\r_u \cS -\l_u \cS^{-p}) (\cS -1)^{-1}(\cS^{p} \r_u-\l_u\cS^{-1} )
			\\
			&&\quad \cdot\prod_{i=1}^{\rightarrow{p-1}} 
			(\cS^{p-i} \r_u-\l_u \cS^{-i})^{-1}(\cS^{p-i} \r_u-\l_u\cS^{-(i+1)} ) \cdot \cS^{-lp+p}
			\cP^{(l,lp)} ,
		\end{eqnarray*}
		where we used formula \eqref{rere}.
		We now recursively apply formula \eqref{imp} for $i$ from $p-1$ to $1$ and
		obtain 
		\begin{eqnarray*}
			&&\quad Q^l=\Re (Q^{l-1})=(\r_u \cS -\l_u \cS^{-p}) (\cS -1)^{-1}(\cS^{p} \r_u-\l_u\cS^{-1} )
			\cS^{-lp+1} \cP^{(l,lp+p-1)} \\
			&&=(\r_u \cS -\l_u \cS^{-p}) (\cS -1)^{-1} \left(\r_{u_p} \cS^{-lp+p+1}
			\cP^{(l,lp+p-1)} -\l_u \cS^{-lp} \cP^{(l,lp+p-1)} \right) \\
			&&= (\r_u \cS -\l_u \cS^{-p}) \cS^{-lp} (\cS-1)^{-1}\left( \r_{u_{(l+1)p}} \cS^{p+1} 
			\cP^{(l,(l+1)p-1)} -\l_{u_{lp}}  \cP^{(l,(l+1)p-1)} \right)\\
			&&= (\r_u \cS -\l_u \cS^{-p}) \cS^{-lp} \cP^{(l+1,(l+1)p-1)} .
		\end{eqnarray*}
		Here we used formula \eqref{prop3} for $l$ being $l+1$ and $k$ being
		$(l+1)p-1$, that is,
		\begin{eqnarray*}
			&&(\cS-1) \cP^{(l+1,(l+1)p-1)}=\r_{u_{(l+1)p}} \cS^{p+1} (\cP^{(l,(l+1)p-1)})
			-\l_{u_{lp}} \cP^{(l,(l+1)p-1)}.
		\end{eqnarray*}
		We completed the induction proof of the statement.
	\end{proof}
	\subsection{Hamiltonian structure}
	The Volterra chain (see Section \ref{ssec:Volterra}) is a special case of Narita-Itoh-Bogoyavlensky lattice for $p=1$. The Hamiltonian formulation of Volterra chain is given by equations \eqref{eq:Volterra_HamStr} and \eqref{eq:VoltHamF}. The natural generalisation to the generic case $p>1$ is
	\begin{equation}
	\label{eq:NIBHam}
	H=\r_u \sum_{i=1}^p \S^i \r_u-\l_u \sum_{i=1}^p \S^{-i} \l_u +\r_u\c_u +\c_u (\S-1)^{-1} \c_u.
	\end{equation}
	\begin{theorem}
		The operator $H$ defined in \eqref{eq:NIBHam} is Hamiltonian. It produces the Narita-Itoh-Bogoyavlensky lattice 
		equation for the Hamiltonian functional
		\begin{equation}
		h=\int\tr u.
		\end{equation}
	\end{theorem}
	\begin{proof}
		The theorem is made of two statements. The second one is a straightforward computation. We have
		\begin{align}
		u_t&=H \frac{\delta h}{\delta u}=H(1)=\sum_{i=1}^p\left(u_iu-uu_{-i}\right)
		\end{align}
		
		The first statement is proved by induction. Let us denote $H_p$ the operator \eqref{eq:NIBHam} when the upper bound of the sums is $p$. We have proved in Proposition \ref{thm:Volterra-Ham} that $H_1$ is Hamiltonian. We show here that, if $H_{p-1}$ is Hamiltonian, $H_p=\r_u\cS^p\r_u-\l_u\cS^{-p}\l_u+H_{p-1}$ is Hamiltonian too. For simplicity we denote $H^p=H_p-H_{p-1}$ the added terms, and use 
		\begin{align}
		P^p&=\int\tr u_pu\theta\theta_p,\\
		P_{p-1}&=\frac12\int\tr\left(2\sum_{k=1}^{p-1}u_ku\theta\theta_k+u^2\theta\theta+(\rho-\lambda)(\rho_1-\lambda_1)\right),\\
		P_p&=P^p+P_{p-1}
		\end{align}
		for the corresponding bivectors. In the remaining part of the proof, we will omit the summation symbol and adopt the convention that any expression including $k$ is summed for $k=1,\ldots ,p-1$.
		
		The Hamiltonian condition reads
		\begin{align}
		\mathbf{pr}_{H_p\theta}(P_p)&=\mathbf{pr}_{H^p\theta}P^p+\mathbf{pr}_{H_{p-1}\theta}P^p+\mathbf{pr}_{H^p\theta}P_{p-1}+\mathbf{pr}_{H_{p-1}\theta}P_{p-1}\\
		&=\mathbf{pr}_{H^p\theta}P^p+\mathbf{pr}_{H_{p-1}\theta}P^p+\mathbf{pr}_{H^p\theta}P_{p-1}=0,
		\end{align}
		where we have dropped the last term because of the inductive hypothesis. The computation of each term follows the lines of the aforementioned proof of Proposition \ref{thm:Volterra-Ham}. Observe that the nonlocal terms are present only in $H_0$ and are the same as in \eqref{eq:Volterra_HamStr}.
		For each of the summands we obtain
		\begin{align}\label{eq:HamNIB1}
		\mathbf{pr}_{H^p\theta}(P^p)&=\int\tr\left(u_p u\theta\theta_p u_p \theta_p-u_pu\theta u\theta\theta_p\right),\\
		\mathbf{pr}_{H_{p-1}\theta}(P^p)&=\int\tr\left(u_{k+p}u_pu\theta\theta_p\theta_{p+k}+u_p^2u\theta\theta_p\theta_p-u_puu_{-k}\theta_{-k}\theta\theta_p\right.\\
		&\qquad\qquad+u_ku\theta\theta_pu_p\theta_k-u_pu\theta\theta_pu_p\theta_p-u_pu_{p-k}\theta_{p-k}u\theta\theta_{p}\\
		&\qquad\qquad-u_pu\theta u\theta\theta_p+u^2\theta\theta_pu_p\theta\\
		&\qquad\qquad\left.+(\rho-\lambda)(uu_{-p}\theta_{-p}\theta-\theta\theta_pu_pu)-[(\cS^p-1)(\rho-\lambda)]u\theta\theta_pu_p\right)\\\intertext{that, using $(\cS^p-1)(\rho-\lambda)=\sum_{k=0}^{p-1}(\theta_ku_k-u_k\theta_k)$ (note the different summation boundary), gives}
		&=\int\tr\left(u_{k+p}u_pu\theta\theta_p\theta_{p+n}+u_p^2u\theta\theta_p\theta_p-u_puu_{-k}\theta_{-k}\theta\theta_p\right.\\\label{eq:HamNIB2}
		&\qquad\qquad\left.-u_pu\theta\theta_pu_p\theta_p+(\rho-\lambda)(uu_{-p}\theta_{-p}\theta-\theta\theta_pu_pu)\right),\\
		\mathbf{pr}_{H^p\theta}(P_{p-1})&=\int\tr\left(u_{p+k}u_ku\theta\theta_k\theta_{p+k}-u_kuu_{-p}\theta_{-p}\theta\theta_k-u^2u_{-p}\theta_{-p}\theta\theta\right.\\
		&\qquad\qquad-u_ku_{k-p}\theta_{k-p}u\theta\theta_k+u_pu\theta u\theta\theta_p+u_pu\theta u_k\theta_k\theta_p\\\label{eq:HamNIB3}
		&\qquad\qquad\left.+(\rho-\lambda)(\theta\theta_pu_pu-uu_{-p}\theta_{-p}\theta)\right).
		\end{align}
		The sum of the three terms, using the properties of the integral, yields
		\begin{equation}
		\mathbf{pr}_{H_p\theta}(P_p)=\int\tr\left(u_pu\theta u_k\theta_k\theta_{p}-u_ku_{k-p}\theta_{k-p}u\theta\theta_k\right).
		\end{equation}
		A similar expression appeared in the computations of \eqref{eq:HamNIB2}. We first shift the second summands by $p-k$, obtaining
		\begin{equation}
		\mathbf{pr}_{H_p\theta}(P_p)=\int\tr u_pu\theta\left( u_k\theta_k-u_{p-k}\theta_{p-k}\right)\theta_{p},
		\end{equation}
		and then observe that the (implicit) summation runs from $1$ to $p-1$, so that the terms in the bracket cancel out when summed over $k$. We can then conclude that $\mathbf{pr}_{H_p\theta}(P_p)=0$, proving our first claim.
	\end{proof}	
	\section{More nonabelian integrable equations}\label{sec:List}
	In this section we present some examples of integrable nonabelian integrable systems.
	They are the nonabelian counterparts of integrable systems in the list in \cite{kmw13}.
	We provide their Lax representations and recursion operators without giving the computational details. Moreover, we are able to obtain the Hamiltonian structure and functional for some of them. 
	To the best of our knowledge, this is the first collection of such systems and their recursion operators are new.

	\subsection{Nonabelian modified Volterra}
	The nonabelian modified Volterra equation 
	\begin{equation}\label{eq:mVol}
	u_t=u_1u^2-u^2u_{-1}
	\end{equation}
	is obtained from the Volterra equation \eqref{eq:Volterra} with the Miura transformation $w=u_1 u$. More precisely, if $u$ is a solution of modified Volterra then $w$ is a solution of Volterra equation. Indeed,
	\begin{equation}
	w_t=w_1w-ww_{-1}=u_2u_1^2u-u_1u^2u_{-1}=u_{1t}u+u_1u_t.
	\end{equation}
	
	Applying the Miura transformation, from the Volterra chain \eqref{eq:Volterra} we obtain the  Hamiltonian structure for \eqref{eq:mVol}
	\begin{eqnarray*}
		&&H=w_*^{-1}
		\left(\r_{u_1u}\cS\r_{u_1u}-\l_{u_1u}\cS^{-1}\l_{u_1u}+\c_{u_1u}\r_{u_1u}+\c_{u_1u}(\cS-1)^{-1}\c_{u_1u}\right)
		{w_*^{\dagger}}^{-1}\\
		&&\quad=\left(\r_u-\l_u\cS^{-1}\right)\left(\cS-1\right)^{-1}\left(\cS\r_u-\l_u\right)+\left(\r_u\cS+\l_{u_1}\right)^{-1}\l_{u_1u}\left(\cS^{-1}\l_u-\r_u\right) 
	\end{eqnarray*}
	with the Hamiltonian functional
	\begin{equation}
	h=\int\tr u_1u.
	\end{equation}
	and the recursion operator $\Re$ 
	\begin{equation}
	\Re=\left(\r_u-\l_u \cS^{-1}\right)(\cS-1)^{-1}\left(\l_{u_1u}-\cS \r_{u_1u}\cS\right)\left(\l_{u_1u}-\r_{u_1u}\cS\right)^{-1}\left(\r_u\cS+\l_{u_1}\right).
	\end{equation}
	Here we used the identity
	\begin{eqnarray*}
		&&\left(\r_u\cS+\l_{u_1}\right)\left(\r_u-\l_u\cS^{-1}\right)=\r_{u_1 u} \cS-\l_{u_1 u}\cS^{-1}
	\end{eqnarray*}
	to simplify the expressions of $H$ and $\Re$.
	\begin{note}
		We can consider the nonabelian version of a more general Volterra-like equation (see \cite{ay94} for the commutative case)
		\begin{equation}
		u_t=u_1 P(u)-P(u)u_{-1},
		\end{equation}
		for $P(u)=\alpha u^2 + \beta u +\gamma$, with $\alpha,\beta,\gamma\in\C$. Similarly to the commutative case, such a system can be reduced to the Volterra one if $\alpha=0$, $\beta\neq 0$ with the change of coordinates $u\mapsto \beta^{-1}(u-\gamma)$; when $\alpha\neq 0$, a change of coordinates $t\mapsto \alpha^{-1} t$, $u\mapsto u-\beta/2\alpha$ produces the same equation with $P(u)=u^2+c$, $c=\gamma/\alpha-\beta^2/(4\alpha^2)$. The Lax pair representation of this system is the same as for the commutative case \cite{kmw13}
		\begin{align}
		U&=\begin{pmatrix}
		c \lambda^{-1} & u \\ -u & \lambda
		\end{pmatrix}& B&=\begin{pmatrix} c^2\lambda^{-2}+uu_{-1} & c\lambda^{-1}u+\lambda u_{-1}\\ -c\lambda^{-1}u_{-1}-\lambda u & \lambda^2+uu_{-1}\end{pmatrix},
		\end{align}
	\end{note}
	\subsection{Nonabelian relativistic Toda}\label{RT}
	The nonabelian version of the relativistic Toda system \cite{Kup00,G11}
	\begin{equation}
	\begin{cases}
	u_t=&u\left(u_{-1}+v\right)-\left(u_1+v_1\right)u\\
	v_t=&vu_{-1}-uv
	\end{cases}
	\end{equation}
	is Hamiltonian with respect to the Hamiltonian operator
	\begin{equation}
	H_1=\begin{pmatrix}
	0 & \l_u-\r_u\cS\\ \cS^{-1}\l_u-\r_u & \r_u\cS-\cS^{-1}\l_u-\c_v
	\end{pmatrix}
	\end{equation}
	and the Hamiltonian functional
	\begin{equation}
	h_0=\int\tr\left(\frac12\left(u^2+v^2\right)+uv+u_1u+uv_1\right).
	\end{equation}
	As for the non-relativistic Toda lattice, the proof that $H_1$ is a Hamiltonian difference operator is obtained by a direct computation, after deriving
	\begin{gather}
	H_1\Theta=\begin{pmatrix}
	u\zeta-\zeta_1u\\u_{-1}\theta_{-1}-\theta u+\zeta_1 u-u_{-1}\zeta_{-1}-\zeta v+v\zeta
	\end{pmatrix},\\
	P=-\int\tr\left(u\zeta\theta+u\theta\zeta_1-u\zeta\zeta_1+v\zeta\zeta\right).
	\end{gather}
	The hierarchy shares the same Lax representation of the commutative case
	\begin{align}
	U&=\begin{pmatrix}
	\lambda v-\lambda^{-1} & u_{-1}\\-1&0
	\end{pmatrix}& B&=\begin{pmatrix}
	-\lambda^{-2}-u_{-1} & \lambda^{-1}u_{-1}\\ -\lambda^{-1} & -u_{-2}-v_{-1}
	\end{pmatrix}
	\end{align}
	from which we can compute the recursion operator
	\begin{equation}
	\begin{split}
	\Re&=\begin{pmatrix}\l_u\r_v-\r_u\cS \l_v & \l_u-\r_u\cS^2\\
	0 & \l_v-\r_v\cS
	\end{pmatrix}\cdot\begin{pmatrix}1 & 0\\
	\left(1-\cS\right)^{-1}\left(\cS^{-1}\l_u-\r_u\cS\right) & 1
	\end{pmatrix}\\
	&\qquad\qquad
	\cdot\begin{pmatrix}
	\left(\l_u-\r_u\cS\right)^{-1} & 0\\0 & \left(1-\cS\right)^{-1}
	\end{pmatrix}
	\end{split}
	\end{equation}
	through the ansatz
	\begin{equation}\label{eq:rTodaAnsatz}
	B^{t}=\lambda^{-2}B^{\tau}+\lambda^{-2}C^{(-2)}+\lambda^{-1}C^{(-1)}+C^{(0)},
	\end{equation}
	where $C^{(-2)}$ and $C^{(0)}$ are diagonal matrices and $C^{(-1)}$ is off-diagonal.
	
	We are able to write down the inverse recursion operator $\Re^{-1}$ in the same way as in the commutative case 
	\cite{kmw13}; it can be obtained in a similar way, replacing the role of $t$ and $\tau$ in the ansatz 
	\eqref{eq:rTodaAnsatz}. We get
	\begin{equation}
	\begin{split}
	\Re^{-1}&=\begin{pmatrix}
	\r_u\cS-\l_u & 0\\
	\left(\cS^{-1}\l_u-\r_u\cS\right) & 1-\cS
	\end{pmatrix}\cdot\begin{pmatrix}
	1 & \left(\r_u\cS \l_v-\l_u\r_v\right)^{-1}\left(\r_u\cS^{2}-\l_u\right)\\
	0 & 1
	\end{pmatrix}\cdot\\
	&\qquad\cdot\begin{pmatrix}
	\left(\r_u\cS \l_v-\l_u\r_v\right)^{-1}& 0 \\0 & \left(\l_v-\r_v\cS\right)^{-1}
	\end{pmatrix}
	\end{split}
	\end{equation}
	The recursion operator $\Re^{-1}$ has a different seed than $\Re$, namely
	\begin{equation}
	\begin{pmatrix}
	u v^{-1}-{v_1}^{-1}u\\
	v_1^{-1}u-u_{-1}v_{-1}^{-1}
	\end{pmatrix}.
	\end{equation}
	Using the recursion operator $\Re$, we obtain the second Hamiltonian structure for the relativistic Toda system, namely we have $H_2=\Re H_1$ and
	\begin{align}
	\left(H_2\right)_{11}&=\l_u\cS^{-1}\l_u-\r_u\cS \r_u-\r_u\c_u+\c_u(1-\cS)^{-1}\c_u\\
	\left(H_2\right)_{12}&=\l_u\r_v-\r_u\c_v-\r_u\cS \r_v + \c_u(1-\cS)^{-1}\c_v\\	
	\left(H_2\right)_{21}&=\l_v\cS^{-1}\l_u-\r_{uv} + \c_v(1-\cS)^{-1}\c_u\\	
	\left(H_2\right)_{22}&= -\c_v\r_v+\c_v(1-\cS)^{-1}\c_v
	\end{align}
	\begin{note}
		In the commutative case, the Relativistic Toda equation \cite{r90}, written in variables
		$\bar{u}$ and $\bar{v}$ is related to the Relativistic Volterra Lattice by the Miura
		transformation $\bar{u}=-vu$ and $\bar{v}=- (u+v_{-1}+1)$ \cite{kmw13}. For the nonabelian case, this transformation produces the \emph{nonabelian Relativistic Volterra lattice}
		\begin{eqnarray*}
			\left\{ {\begin{array}{l} u_t= vu(1+u)-uv_{-1}(1+u_{-1})\\ v_t=  (1+v_1)u_1v-(1+v)vu. \end{array} } \right.
		\end{eqnarray*}
		It possesses a Lax representation
		\begin{eqnarray*}
			U= \left( \begin{array}{cc} \lambda^2+2u_1+2v+1 & -\lambda(2v+1)-\lambda^{-1}(2u_1+1) \\ -\lambda(2u_1+1)-\lambda^{-1}(2v+1) & \lambda^{-2}+2u_1+2v+1 \end{array}\right)\\
			B= \left( \begin{array}{cc}-\frac{\lambda^2-\lambda^{-2}}{8}+vu+\frac u2+\frac v2 & \frac\lambda4(2v+1)+\frac{\lambda^{-1}}{4}(2u+1) \\ \frac\lambda4(2u+1)+\frac{\lambda^{-1}}{4}(2v+1) & \frac{\lambda^2-\lambda^{-2}}{8}+vu+\frac u2+\frac v2 \end{array}\right) 
		\end{eqnarray*}
	\end{note}
	\subsection{Nonabelian Merola-Ragnisco-Tu Lattice}
	\begin{eqnarray*}
		\left\{ {\begin{array}{l} u_t= u_{1}-u v u\\ v_t=  -v_{-1}+v uv \end{array} } \right.
	\end{eqnarray*}
	Lax representation:
	\begin{eqnarray*}
		U= \left( \begin{array}{cc} -1 & v \\ u & -2\lambda-u v \end{array}\right); \qquad
		B= \left( \begin{array}{cc}-\lambda & -v_{-1} \\ -u & \lambda \end{array}\right) 
	\end{eqnarray*}
	Hamiltonian structure:
	\begin{eqnarray*}
		H= \left( \begin{array}{cc} 0 & 1 \\ -1 & 0 \end{array}\right); \qquad
		h= u_1v-\frac{uvuv}{2}
	\end{eqnarray*}
	Recursion operator:
	\begin{eqnarray*}
		\Re= \left( \begin{array}{cc} \cS-\l_{uv} & 0 \\ 0 & \cS^{-1}-\r_{uv} \end{array}\right)+\left( \begin{array}{cc} -\r_u\cS & -\l_u \\ l_v & \r_v\cS \end{array}\right)(\cS-1)^{-1}\left( \begin{array}{cc} \r_v & \l_u \\ \l_v & \r_u \end{array}\right)
	\end{eqnarray*}
	\subsection{Nonabelian Kaup Lattice} \label{ssec:Kaup}
	\begin{eqnarray*}
		\left\{ {\begin{array}{l} u_t=(u_1-u) (u+v) \\ v_t= (u+v) (v-v_{-1}) \end{array} } \right.
	\end{eqnarray*}
	Lax representation 
	\begin{eqnarray*}
		U= \left( \begin{array}{cc} u-\lambda & uv+\lambda (u+v)+\lambda^2 \\ 1 & v-\lambda \end{array}\right); \qquad 
		B= \left( \begin{array}{cc}u & (u+\lambda)(v_{-1}+\lambda) \\ 1 & v_{-1} \end{array}\right)
	\end{eqnarray*}
	We compute its recursion operator with the ansatz
	\begin{equation*}
	B^{(t)}= \lambda B^{(\tau)}+\left( \begin{array}{cc}\lambda a+e & \lambda^2a+\lambda b+c  \\a& \lambda a+d \end{array}\right)
	\end{equation*}
	and obtain 
	\begin{eqnarray*}
		&\Re=\left( \begin{array}{cc}  -\r_u\cS & -\l_u\\ -\r_v&-\l_v\cS^{-1}\end{array}\right)+\left( \begin{array}{c}  \l_{u_1-u}\r_{u+v}\cS\\ \l_{u+v}\r_{v-v_{-1}}\end{array}\right)\left(\l_{u+v}-\r_{u+v}\cS\right)^{-1}\left( \begin{array}{cc} 1& 1\end{array}\right)\\
		&+\left(
		\begin{array}{c}  \cS\\ -1\end{array}\right)\left(\cS-1\right)^{-1}
		\left( \begin{array}{cc} \r_v\left(1-\cS\right)& \l_u\left(1-\cS^{-1}\right)\end{array}\right)
	\end{eqnarray*}
	\subsection{Nonabelian Ablowitz-Ladik Lattice} \label{ssec:AblLad}
	\begin{equation}
		\left\{ {\begin{array}{l} u_t=\alpha(u_1-u_1vu)+\beta(uvu_{-1}-u_{-1}) \\ v_t= \alpha(v u v_{-1}-v_{-1})+\beta(v_1-v_1uv) \end{array} } \right.\qquad\alpha,\beta\in\C
	\end{equation}
	Its Lax representation is
	\begin{eqnarray*}
		U= \left( \begin{array}{cc} \lambda & u \\ v & \lambda^{-1} \end{array}\right); \qquad 
		B=\alpha \left( \begin{array}{cc}\lambda^2-u v_{-1} &  \lambda u  \\ \lambda v_{-1} & 0 \end{array}\right)+\beta \left( \begin{array}{cc}0 &  \lambda^{-1} u_{-1}  \\ \lambda^{-1} v & \lambda^{-2}-vu_{-1} \end{array}\right)
	\end{eqnarray*}
	We compute its recursion operator
	with the ansatz
	\begin{equation*}
	B^{t}= \lambda^2 B^{\tau}+\left( \begin{array}{cc}a & \lambda b \\\lambda c& \lambda^2 d \end{array}\right)
	\end{equation*}
	and obtain 
	\begin{eqnarray*}
		&&\Re=\left( \begin{array}{cc}  \cS & 0\\ 0&\cS^{-1} \end{array}\right) +\left( \begin{array}{c} \r_u \cS \\ -\l_v 
		\end{array}\right)(1- \cS)^{-1} 
		\left( \begin{array}{cc} \r_{v}\cS &\l_u \cS^{-1} \end{array}\right)\\ 
		&&\qquad+\left( \begin{array}{c} \r_{vu-1} \cS \l_u\\ -\l_{vu-1} \r_{v_{-1}} \end{array}\right)(\l_{vu-1}-\r_{vu-1} 
		\cS)^{-1}
		\left( \begin{array}{cc} \l_v &\r_u  \end{array}\right) 
	\end{eqnarray*}
	
	We are able to write down the inverse recursion operator $\Re^{-1}$ in the same way as in the commutative case 
	\cite{kmw13}; it can be obtained in a similar way, replacing the role of $t$ and $\tau$ in the ansatz. We get
	\begin{eqnarray*}
		&&\Re^{-1}=\left( \begin{array}{cc}  \cS^{-1} & 0\\ 0&\cS \end{array}\right) +\left( \begin{array}{c} \l_u  \\ -\r_v \cS
		\end{array}\right)(\cS-1)^{-1} 
		\left( \begin{array}{cc} \l_{v}\cS^{-1} &\r_u \cS \end{array}\right)\\ 
		&&\qquad+\left( \begin{array}{c} -\l_{uv-1} \r_{u_{-1}} \\ \r_{uv-1} \cS \l_v \end{array}\right)(\l_{uv-1}-\r_{uv-1} 
		\cS)^{-1}
		\left( \begin{array}{cc} \r_v &\l_u  \end{array}\right) 
	\end{eqnarray*}
	Note that the seed for $\Re$ is the coefficient of $\alpha$ in the equation while the seed for $\Re^{-1}$ is the coefficient
	of $\beta$.
	\subsection{Nonabelian Chen-Lee-Liu lattice}\label{ssec:cll}
	\begin{eqnarray*}
		\left\{ {\begin{array}{l} u_t=(u_1-u)(1+vu) \\ v_t= (1+vu)(v-v_{-1}) \end{array} } \right.
	\end{eqnarray*}
	Its Lax representation is
	\begin{eqnarray*}
		U= \left( \begin{array}{cc}\lambda+ uv & u\\(1-\lambda)v& 1 \end{array}\right); \qquad 
		B= \left( \begin{array}{cc}\lambda-1+uv_{-1} & u  \\ (1-\lambda) v_{-1} & 0 \end{array}\right)
	\end{eqnarray*}
	We compute its recursion operator
	with the ansatz
	\begin{equation*}
	B^{(t)}= \lambda B^{(\tau)}+\left( \begin{array}{cc}a & b \\(1-\lambda)c& (1-\lambda) d \end{array}\right)
	\end{equation*}
	and obtain 
	\begin{eqnarray*}
		&&\Re=\left( \begin{array}{cc} \r_{1+vu} \cS-\r_{vu}+\l_{(u_1-u)v} & \l_{u_1-u} \r_u +\l_u \r_u \cS^{-1}\\ \l_v \r_v& 
			\cS^{-1} \end{array}\right)\\&&\qquad +\left( \begin{array}{c} \r_u\\ -\l_v \end{array}\right)(1- \cS)^{-1} 
		\left( \begin{array}{cc} \r_{v}(1-\cS) &\l_u (1-\cS^{-1}) \end{array}\right)\\ 
		&&\qquad+\left( \begin{array}{c} \l_{(u-u_1)(1+vu)}\\ \l_{1+vu} \r_{v_{-1}-v} \end{array}\right)(\l_{1+vu}-\r_{1+vu} 
		\cS)^{-1}
		\left( \begin{array}{cc} \l_v &\r_u  \end{array}\right) 
	\end{eqnarray*}
	
	\subsection{Nonabelian Blaszak-Marciniak Lattice}\label{BM}
	\begin{eqnarray*}
		\left\{ {\begin{array}{l} u_t=w_1 -w_{-1} \\ v_t= w_{-1}u_{-1}-u w \\ w_t=wv-v_1w \end{array} } \right.
	\end{eqnarray*}
	It is a Hamiltonian system with a Hamiltonian structure
	\begin{equation}
	H_1=\begin{pmatrix}
	\cS-\cS^{-1} & 0 & 0 \\ 0 & \l_v-\r_v & \cS^{-1}\l_w-\r_w \\ 0 & \l_w-\r_w\cS & 0
	\end{pmatrix},\qquad h=uw+\frac{1}{2}v^2
	\end{equation}
	Its Lax representation is
	\begin{eqnarray*}
		L = \cS^2 + u_1 \cS - v_1 + w\cS^{-1},
		\qquad A = \cS^2 + u_1 \cS - v_1
	\end{eqnarray*}
	from which we can compute the recursion operator $\Re$ by 
	$$L_t=[L, \ A^{t}]$$
	with the ansatz
	\begin{equation*}
	A^{t}=LA^{\tau}+a \cS+b+c \cS^{-1}.
	\end{equation*}
	The explicit form of the operator is rather big. In factorised form we have $\Re=KH_1^{-1}$ with
	\begin{align}
	\left(K\right)_{11}&=\l_v\cS^{-1}-\cS \r_v+ (\l_u \cS-\r_{u}) (1-\cS^2)^{-1} (\cS \r_u-\l_u)\\
	\left(K\right)_{21}&=\cS^{-1} \l_w\cS^{-1}-\r_w \cS - \c_{v} (1-\cS^2)^{-1} (\cS \r_u-\l_u)\\
	\left(K\right)_{22}&=\l_u \r_w\cS-\cS^{-1} \l_w \r_u+\r_v\c_v - \c_v(1-\cS^2)^{-1}\c_v\\
	\left(K\right)_{31}&= (\l_w-\r_w \cS) (1-\cS^2)^{-1} (\cS \r_u-\l_u)\\
	\left(K\right)_{32}&=\r_w\cS \r_v-\l_w \r_v + (\l_w-\r_w \cS) (1-\cS^2)^{-1} \c_v\\
	\left(K\right)_{33}&=\r_w\cS\r_u-\l_w \r_w + (\l_w-\r_w \cS) (1-\cS^2)^{-1} (\r_w-\cS^{-1} \l_w) .
	\end{align}

	\section{Discussion and further work}

	In this paper we have computed nonabelian Hamiltonian structures for some difference system. It should be noted that in the scalar cases (Volterra and modified Volterra, as well as Narita-Itoh-Bogoyavlensky) all the Hamiltonian operators we have identified are nonlocal, even when the nonlocal terms do not contribute to the equation. On the other hand, the first Hamiltonian structures of the two Toda systems we have investigated are local.  We could not find local nonabelian Hamiltonian structures for the systems \ref{ssec:Kaup}, \ref{ssec:AblLad}, \ref{ssec:cll}, despite their commutative counterparts are extremely simple. For instance, the Hamiltonian structure for the Abelian Ablowitz-Ladik lattice is \cite{kmw13}
	\begin{equation}
	H=\begin{pmatrix}0 & 1-uv \\-(1-uv) & 0
	\end{pmatrix}
	\end{equation}
	A recent development in the study of nonabelian ODEs \cite{ar15} suggests that brackets of Loday type can play the same role as Hamiltonian structures, satisfying weaker conditions.

	A wider investigation (and, possibly, classification) of the nonabelian difference Hamiltonian structures (on the lines of what has been done for the commutative case in \cite{dskvw19}) will be carried out in a forthcoming work, where we will also present the nonlocal structures for the aforementioned systems \ref{ssec:Kaup}, \ref{ssec:AblLad}, and \ref{ssec:cll}.
	
	The generalisation from the Abelian case to the nonabelian case is not straightforward. 
	In Abelian case, the Narita-Itoh-Bogoyavlensky lattice \eqref{bogcom} has a product form
	\begin{equation}
	v_t=v (\prod _{k=1}^p v_{k}-\prod_{k=1}^p v_{-k}), \label{bogpro}
	\end{equation}
	For fixed~$p$, it transforms into \eqref{bogcom} under the
	transformation $u= \prod_{k=0}^{p-1} v_{k}$. 
	There is also the modified Bogoyavlensky chain given by
	\begin{equation}\label{bogmod}
	w_t=w^2 (\prod_{k=1}^p w_{k}-\prod_{k=1}^p w_{-k}),
	\end{equation}
	which is related to \eqref{bogcom} by the Miura transformation $u=\prod_{k=0}^p
	w_{k}$. There Miura transformations are not valid for the nonabelian case except when $p=1$.
	
	In 2011, Adler and Postnikov introduced a family of integrable lattice hierarchies associated with fractional
	Lax operators in \cite{adler1, adler2}. One simple
	example is
	\begin{equation}\label{adler}
	u_t=u^2 (u_p \cdots u_1-u_{-1}\cdots u_{-p})-u (u_{p-1}\cdots u_1-u_{-1}\cdots
	u_{1-p}), \quad 2\leq p\in\bbbn ,
	\end{equation}
	which is an integrable discretisation for the Sawada-Kotera
	equation.  Notice that equation \eqref{adler} is a combination of equations
	\eqref{bogmod} and \eqref{bogpro} with different $p$.  For them, there are no direct generalisation to nonabelian 
	case.
	In differential case, the Sawada-Kotera
	equation possesses no  nonabelian version \cite{OS98, ow00}.
	It would be interesting to see whether there exists nonabelian discretisation for the Sawada-Kotera
	equation. On the other hand, there are two nonabelian discretizations of Burgers' equation:
	\begin{eqnarray*}
		u_t=(u_1-u)u\quad \mbox{and} \quad u_t=u(u_1-u),
	\end{eqnarray*}
	which can be transformed, by the Cole-Hopf transformation 
	\begin{eqnarray*}
		u=v_1 v^{-1}\quad \mbox{and} \quad u=v^{-1} v_1
	\end{eqnarray*}
	respectively, into the linear equation $v_t=v_1$.
	
	In this paper, we present some new examples of nonabelian integrable differential-difference systems with their 
	Lax representations, Hamiltonian and recursion operators. The list is by no mean to be complete.  
	For example, we didn't include the nonabelian Belov-Chaltikian 
	Lattice
	\begin{eqnarray*}
		\left\{ {\begin{array}{l} u_t=v_2 u-u v_{-1} \\ v_t= v_1  v-vv_{-1}+u_{-1} -u \end{array} } \right.
	\end{eqnarray*}
	whose Lax representation is
	\begin{eqnarray*}
		U= \left( \begin{array}{ccc} \lambda & \lambda v & \lambda u_{-1} \\ 1 & 0 & 0 \\ 0 & 1 & 0 \end{array}\right); \qquad 
		B= \left( \begin{array}{ccc}v-\lambda & - \lambda v & -\lambda u_{-1}  \\ -1 & v_{-1} & 0 \\ 0 & -1 & v_{-2} \end{array}\right).
	\end{eqnarray*}
	In the Abelian case, it is the Boussinesq lattice related to the lattice $W_3$-algebra \cite{hiin97}, and can be generalised to an $m$-component  Boussinesq lattice related to the lattice
	$W_{m+1}$-algebra \cite{beffawang12}. It deserves a further study on its own right. The classification problem even for scalar case remains open.
	
	The recursion operators for nonabelian integrable differential-difference systems are highly nonlocal. For Abelian case, there are some general results to prove that they generate local hierarchies
	under checkable conditions, see for example \cite{wang09}. In section \ref{sec4}, we proved that this is indeed true for 
	the nonabelian Narita-Itoh-Bogoyavlensky lattice. 
	However, we haven't proved any general results for other integrable equations.
	\section*{Acknowledgements}
	The paper is supported by the EPSRC grant
	EP/P012698/1. Both authors gratefully acknowledge the financial support. JPW would like to thank V.~V.~Sokolov for useful discussions.

\end{document}